\documentclass[letterpaper,12pt]{extarticle}
\usepackage{tabularx} 
\usepackage{preamble}
\usepackage{packages}
\usepackage{breqn}
\usepackage{enumerate}
\usepackage[font=small,labelfont=bf]{caption}
\usepackage[left=30mm,right=30mm,top=30mm,bottom=30mm]{geometry}
\usepackage{pgfplots}
\usetikzlibrary{calc}
\usepackage{extarrows}
\usepackage{amsmath,amssymb}
\usepackage{amsfonts}
\usepackage{booktabs,subcaption,amsfonts,dcolumn}
\newcolumntype{d}[1]{D..{#1}}
\definecolor{refkey}{rgb}{0.9451,0.2706,0.4941}
\definecolor{labelkey}{rgb}{0.9451,0.2706,0.4941}
\DeclareFontFamily{U}{rcjhbltx}{}
\DeclareFontShape{U}{rcjhbltx}{m}{n}{<->rcjhbltx}{}
\DeclareSymbolFont{hebrewletters}{U}{rcjhbltx}{m}{n}
\DeclareMathSymbol{\tet}{\mathord}{hebrewletters}{84}
\DeclareMathSymbol{\pey}{\mathord}{hebrewletters}{112}

\def\z2{$\mathbb{Z}_2$}

\definecolor{darkgray}{rgb}{0.33, 0.33, 0.33}
\usepackage{braket}

\newcommand{\Tr}{{\mathrm{Tr}}}

\usepackage{physics}
\usepackage{bbm}
\usepackage{amsthm}
\usepackage{mathrsfs}
\renewcommand{\arraystretch}{1.3}
\newcommand{\id}{\mathbbm{1}}

\theoremstyle{definition}
\newtheorem{theorem}{Theorem}[section]
\newtheorem{proposition}[theorem]{Proposition}
\newtheorem{lemma}[theorem]{Lemma}
\newtheorem{corollary}[theorem]{Corollary}
\newtheorem{definition}[theorem]{Definition}

\newtheorem{example}[theorem]{Example}
\theoremstyle{remark}
\newtheorem{remark}[theorem]{Remark}

\usepackage{authblk}
\begin{document}

\title{{\LARGE \bf Band Flattening and Overlap Fermion}}
\author[1]{Taro Kimura}
\affil[1]{{\small Institut de Math{\'e}matiques de Bourgogne, Universit{\'e} de Bourgogne, CNRS, France}}
\author[2]{Masataka Watanabe}
\affil[2]{{\small Graduate School of Informatics, Nagoya University, Nagoya 464-8601, Japan}}
\date{}

\maketitle
\begin{abstract}
We show that, for each symmetry class based on the tenfold way classification, the effective Dirac operator obtained by integrating out the additional bulk direction takes a value in the corresponding classifying space, from which we obtain the flat band Hamiltonian.
We then obtain the overlap Dirac operator for each symmetry class and establish the Ginsparg--Wilson relation associated with $\mathcal{C}$ and $\mathcal{T}$ symmetries, and also the mod-two index theorem.
\end{abstract}

\tableofcontents

\section{Introduction}

Study of topological phases of matter, which has been originated in condensed-matter physics, now provides an interdisciplinary arena of research involving various domains of theoretical and experimental physics and also mathematics.
A topological insulator phase is a primary example of the topological phases, that exhibits a gapless surface state, while the interior behaves as a gapped insulator.
This gapless surface state is topologically protected, and the topological insulator cannot be transformed continuously to a topologically trivial insulator.
For a topological band insulator, one can consider the topological invariant associated with its band structure, which plays an essential role in characterization of topological property.
For example, the TKNN number~\cite{Thouless:1982zz,Kohmoto:1985AP,Niu:1984uz} is the topological invariant associated with the two-dimensional class A system, which is given by integrating the Berry curvature (the first Chern class of the Bloch bundle) over the two-dimensional Brillouin torus.

Since the topological property does not depend on the detail of the band structure, we often use the band flattened system to simplify the argument (see, e.g., \cite{Schnyder:2008tya,Kitaev:2009mg}).
In this paper, we provide a systematic methodology, that we call the bulk extension, to obtain the band flattened Hamiltonian from a generic gapped Hamiltonian of free fermion.
Here is the summary of the prescription.
\begin{enumerate}
    \item Consider a $d$-dimensional gapped free fermion Hamiltonian $H$.
    \item Construct a $(d+1)$-dimensional Dirac operator $D$ by adding an extra direction.
    \item Compute the functional determinant $\det D$ while imposing the periodic boundary condition in the extra direction together with the Pauli--Villars regulator.
    \item Read off the effective Dirac operator $\overline{D}$ from the determinant, and convert it to the band flattened Hamiltonian in the bulk limit, $\overline{H} = H/\sqrt{H^2} =: \operatorname{sgn}(H)$.
\end{enumerate}
We apply this formalism to generic symmetry classes of topological insulators and superconductors~\cite{Schnyder:2008tya,Kitaev:2009mg} based on the Altland--Zirnbauer (AZ) tenfold way classification~\cite{Altland:1997zz}. 
Here is the first result of this paper.
\begin{theorem}[Theorem~\ref{thm:band_flatten_Ham}]
Let $\gamma$ be (one of) the mass matrix of the gapped Hamiltonian $H$. 
Then, we obtain the band flattened Hamiltonian from the effective Dirac operator under the periodic boundary condition,  $\overline{H} = \gamma \overline{D}$, in the bulk limit.
\end{theorem}

In fact, for class $\mathscr{C}$ system, the effective Dirac operator $\overline{D}$ takes a value in the classifying space ${S}_{\mathscr{C}} \in \{C_{0,1}, R_{0,\ldots,7} \}$ (Proposition~\ref{prop:V_classifying_space}), which plays an essential role in the tenfold way classification: 
The homotopy group of ${S}_{\mathscr{C}}$ characterizes topological property of the topological insulator/superconductor~\cite{Schnyder:2008tya,Kitaev:2009mg}.
See Table~\ref{tab:classification}.

The formalism of the bulk extension presented in this paper is motivated by the overlap Dirac operator, showing an exact chiral symmetry on a lattice, that was originally formulated for the class A system~\cite{Neuberger:1997fp,Neuberger:1997bg,Neuberger:1998wv}.
We extend the original construction of the overlap operator to generic AZ tenfold way symmetry classes.
\begin{proposition}[Proposition~\ref{lemma:eff_Dirac_bulk_lim}]
The overlap Dirac operator of class $\mathscr{C}$ is given by
\begin{align}
    D_{\text{ov}} = \frac{1}{2} (1 + V)
    \, , \qquad 
    V \in S_{\mathscr{C}}
    \, .
\end{align}
\end{proposition}

It has been known~\cite{Neuberger:1997fp,Neuberger:1997bg,Neuberger:1998wv,Hasenfratz:1998ri,Luscher:1998pqa} that the overlap Dirac operator of class A obeys Ginsparg--Wilson (GW) relation~\cite{Ginsparg:1981bj} (recovering dependence on the lattice constant $a$),%
\footnote{%
Another realization of GW relation is achieved by the perfect action~\cite{Hasenfratz:1993sp}.
}
\begin{align}
    \gamma D + D \gamma = a D \gamma D
    \, ,
\end{align}
which is interpreted as a non-linear deformation of the chiral symmetry relation, $\{\gamma, D\} = \gamma D + D \gamma = 0$.
Applying the same argument to $\mathcal{C}$ and $\mathcal{T}$ symmetries, we obtain the corresponding GW relation.
\begin{theorem}[Theorem~\ref{thm:CT-GW_relation}]
Let $C$ and $T$ be the unitary operators defined in~\eqref{eq:CT_Ham_action}.
For the system with $\mathcal{C}$ and $\mathcal{T}$ symmetries, the overlap Dirac operator obeys,
\begin{align}
    C D + D^{\text{T}} C = a D^{\text{T}} C D
    \, , \qquad 
    T D + D^* T = a D^* T D
    \, .
\end{align}
\end{theorem}

These relations immediately imply an anomaly under $\mathcal{C}$ and $\mathcal{T}$ transformations similarly to the parity anomaly~\cite{Bietenholz:2000ca} in the overlap formalism, which is related to the anomalous behavior of Majorana(--Weyl) fermion (hence, $\mathcal{C}$ transformation)~\cite{Huet:1996pw,Narayanan:1996mr,Inagaki:2004ar,Suzuki:2004ht,Hayakawa:2006fd}, and of the $\mathcal{T}$-invariant topological system~\cite{Fukui:2009pc,Ringel:2012fm}.
We remark that GW relation with respect to an additional symmetry has been also discussed in~\cite{Kimura:2015ixh} in the context of topological crystalline insulators/superconductors~\cite{Ando:2015sia}.

The overlap formalism also provides a concise description of the index theorem.
It has been established that the $\mathbb{Z}$-valued index of the overlap Dirac operator, $\operatorname{ind}(D_{\text{ov}}) = \dim \operatorname{ker}D_{\text{ov}} - \dim \operatorname{coker} D_{\text{ov}}$, is given as follows.
\begin{proposition}[Hasenfratz--Laliena--Niedermayer~\cite{Hasenfratz:1998ri}, Lüscher~\cite{Luscher:1998pqa}, Adams~\cite{Adams:1998eg}]\label{prop:ind_thm_Z}
Let $\eta(A)$ be the eta invariant of a self-adjoint operator $A$.
Then, the $\mathbb{Z}$-valued index of the overlap Dirac operator is given by
\begin{align}
    \operatorname{ind}(D_{\text{ov}}) = - \frac{1}{2} \tr \operatorname{sgn} H = - \frac{1}{2} \eta(H) 
    \, .
\end{align}
\end{proposition}

This index agrees with the bulk topological invariant associated with the gapped Hamiltonian $H$~\cite{Schnyder:2008tya,Kitaev:2009mg}.
In fact, this agreement also holds for the $\mathbb{Z}_2$-topological invariant and the mod-two index of the overlap Dirac operator.
\begin{theorem}[Theorem~\ref{thm:Z2_index}]
The mod-two index of overlap Dirac operator, $\nu = \operatorname{ind}(D_{\text{ov}}) = \dim \operatorname{ker}(D_{\text{ov}})$, is given by
\begin{align}
    (-1)^{\nu} = \det V \, .
\end{align}
\end{theorem}

The use of determinant signature to define the mod-two index has been proposed specifically for (8$n$+2)-dimensional Majorana--Weyl fermion~\cite{Huet:1996pw,Narayanan:1996mr}.
We remark that the mod-two index has been recently formulated in the domain-wall fermion formalism~\cite{Fukaya:2020tjk}, which has a similar expression using the sign factor appearing in the Dirac operator determinant.
See also recent approaches to the index theorem on a lattice~\cite{Yamashita:2020nkf,Kubota:2020tpr}.

\subsubsection*{Organization of the paper}
The remaining part of this paper is organized as follows.
In Sec.~\ref{sec:preliminaries}, we discuss preliminary facts, including the relation between Hamiltonian formalism and Lagrangian formalism, and the symmetry classification.
In Sec.~\ref{sec:bulk_ext}, we apply the formalism, that we call the bulk extension, to obtain the band flattened Hamiltonian.
For each symmetry class, we prove that the effective Dirac operator takes a value in the corresponding classifying space.
In Sec.~\ref{sec:GW_relation}, we explore the overlap Dirac operator obtained through the bulk extension with the open boundary condition.
We prove that the overlap operator obeys GW relation with respect to $\mathcal{C}$ and $\mathcal{T}$ symmetries, and discuss the anomalous behavior under $\mathcal{C}$ and $\mathcal{T}$ transformations.
We also establish the mod-two index of the overlap Dirac operator.

\subsubsection*{Acknowledgements}
We would like to thank Mikio Furuta for insightful comments on the preliminary version of the draft.
The work of TK was in part supported by EIPHI Graduate School (No.~ANR-17-EURE-0002) and Bourgogne-Franche-Comté region.
MW is supported in part by Grant-in-Aid for JSPS Fellows (No.~22J00752).

\subsubsection*{Note added}
While completing this manuscript, we became aware of a recent preprint by Clancy--Kaplan--Singh~\cite{Clancy:2023ino} also addressing the overlap fermion associated with $\mathcal{C}$ and $\mathcal{T}$ symmetries, and the mod-two index discussed in Sec.~\ref{sec:GW_relation}.

\section{Preliminaries}\label{sec:preliminaries}

\subsubsection*{Notations}

\begin{itemize}
    \item 
    For $x \in \mathbb{K}$, let $x^*$ be its $\mathbb{K}$-conjugate.
    We denote the conjugate matrix of $M$ by $M^\dag := M^{*\text{T}}$.
    We define the set of self-conjugate matrices (real symmetric for $\mathbb{R}$, complex hermitian for $\mathbb{C}$, quaternion self-dual for $\mathbb{H}$) of size $n$ by
\begin{align}
    \mathsf{H}(n,\mathbb{K}) = \{ M \in \mathbb{K}^{n \times n} \mid M^\dag = M \}
    \, .
\end{align}

    \item 
    We define the set of skew-conjugate matrices of size $n$ by
\begin{align}
    \widetilde{\mathsf{H}}(n,\mathbb{K}) = \{ M \in \mathbb{K}^{n \times n} \mid M^\dag = - M \}
    \, .
\end{align}

    \item We denote a compact symplectic group by $\mathrm{Sp}(n) = \mathrm{Sp}(2n,\mathbb{C}) \cap \mathrm{U}(2n)$.

    \item We denote a commutator and an anti-commutator by $[a, b] = ab - ba$ and $\{ a, b \} = ab + ba$.

    \item We denote $\mathbb{Z}_n = \mathbb{Z}/n \mathbb{Z}$.
\end{itemize}

\subsection{Lagrangian vs Hamiltonian}\label{sec:Lagrangian_Hamiltonian}

Let $d$ be the spacial dimension, and the spacetime dimension $d+1$.
Let $\{ \gamma_\mu \}_{\mu=0,\ldots,d}$ be the Euclidean gamma matrices, which are hermitian and obey the relation $\{\gamma_\mu, \gamma_\nu \} = \gamma_\mu \gamma_\nu + \gamma_\nu \gamma_\mu = 2 \delta_{\mu,\nu}$.
The free Dirac Lagrangian in the $(d+1)$-dimensional Euclidean spacetime is given by
\begin{align}
    \mathscr{L} = \bar{\psi} \qty( \gamma^\mu \partial_\mu + m ) \psi =: \bar{\psi} D \psi
    \, ,
    \qquad
    D = \gamma^\mu \partial_\mu + m
    \, ,
    \label{eq:HvsD}
\end{align}
where the associated Dirac operator $D$ is non-hermitian in general.%
\footnote{%
The Dirac operator becomes hermitian in the Lorentzian signature.
}
In fact, non-hermitian Hamiltonian discussed in, e.g., \cite{Lee:2019ole}, has a direct interpretation as a Dirac operator.
We remark that the Dirac operator becomes anti-hermitian in the case $m = 0$, $D^\dag = - D$. 
Considering the 0-direction as a ``time'' direction, and putting $\bar\psi = \psi^\dag \gamma_0$, we then obtain the hermitian Hamiltonian as follows:
\begin{subequations}
\begin{align}
    \mathscr{L} & = \psi^\dag \qty( \partial_0 + \gamma_0 \vec{\gamma}\cdot\vec{\partial} + m \gamma_0) \psi
    =: \psi^\dag \qty( \partial_0 + H ) \psi
    \, , \\
    \mathscr{H} & = \psi^\dag H \psi 
    = \psi^\dag \left( \gamma_0 \vec{\gamma} \cdot \vec{\partial} + m \gamma_0 \right) \psi 
    = \psi^\dag \left( - i \vec{\tilde{\gamma}} \cdot \vec{\partial} + m \gamma_0 \right) \psi 
\end{align}
\end{subequations}
where $\tilde{\gamma}_j = i \gamma_0 \gamma_j$ is a hermitian gamma matrix for $j = 1,\ldots,d$.
Hence, we have the relation between the Dirac operator and the Hamiltonian,
\begin{align}
    D = \gamma_0 H + \gamma_0 \partial_0
    \, .
    \label{eq:Dirac_from_Hamiltonian}
\end{align}
In other words, we may identify the mass matrix in the Hamiltonian with the zero-th gamma matrix $\gamma_0$, so that the mass term is proportional to the identity matrix in the Dirac operator.

\begin{example}[$d = 2$]
Let $\{ \sigma_i \}_{i=1,2,3}$ be the Pauli matrices.
The momentum space representation of the massive Dirac Hamiltonian of class A in $d = 2$ is given by
\begin{align}
    H = p_1 \sigma_1 + p_2 \sigma_2 + m \sigma_3
    \, .
\end{align}
In this case, we identify the mass matrix, $\gamma_0 = \sigma_3$.
The corresponding Dirac operator in $2+1$ dimensions is given by
\begin{align}
    D = p_1 (\sigma_3 \sigma_1) + p_2 (\sigma_3 \sigma_2) + i p_0 \sigma_3 + m
    = i p_1 \sigma_2 - i p_2 \sigma_1 + i p_0 \sigma_3 + m
    \, ,
\end{align}
which is not hermitian.
In the massless case $m = 0$, $D^\dag = - D$, and it shows the parity symmetry.
\end{example}

\begin{example}[$d=3$]
The massive Dirac Hamiltonian of class A in $d = 3$ is given as follows:
\begin{align}
    H = \vec{p} \cdot (\vec{\sigma} \otimes \sigma_3) + m (\id \otimes \sigma_2)
    =
    \begin{pmatrix}
    \vec{p} \cdot \vec{\sigma} & -im \\ +im & - \vec{p} \cdot \vec{\sigma}
    \end{pmatrix}
    \, .
\end{align}
Then, having the mass matrix $\gamma_0 = \id \otimes \sigma_2$, the Dirac operator is given by
\begin{align}
    D & = - i \vec{p} \cdot (\vec{\sigma} \otimes \sigma_1) + i p_0 (\id \otimes \sigma_2) + m (\id \otimes \id)
    = 
    \begin{pmatrix}
    m & + p_0 - i \vec{p} \cdot \vec{\sigma} \\ - p_0 - \vec{p} \cdot \vec{\sigma} & m
    \end{pmatrix}
    \, .
\end{align}
If $m = 0$, it shows the chiral symmetry $\{ D, \Gamma \} = 0$ where $\Gamma = \id \otimes \sigma_3$.    
\end{example}

\subsection{Symmetry and classification}

Let us introduce the discrete symmetries, $\mathcal{C}$ and $\mathcal{T}$, which play an essential role in the classification of Hamiltonian and Dirac operator.

\begin{definition}\label{def:CT_Ham_action}
Let $C$ and $T$ be unitary operators, which act on a Hamiltonian as follows,
\begin{align}
    C H C^{-1} = - H^*
    \, , \qquad 
    T H T^{-1} = + H^*
    \, .
    \label{eq:CT_Ham_action}
\end{align}
In the momentum space representation, we have $C H(p) C^{-1} = - H(-p)^*$ and $T H(p) T^{-1} = + H(-p)^*$.
We define the complex conjugation operator $K$, $K X K = X^*$ for any operator $X$.
Then, we define anti-unitary operators, that we call charge conjugation operator $\mathcal{C}$ and time reversal operator $\mathcal{T}$,
\begin{align}
	\mathcal{C} = CK
	\, , \qquad
	\mathcal{T} = TK
	\, .
\end{align}    
If there exist $\mathcal{C}$ and $\mathcal{T}$ operators for a given Hamiltonian $H$, we say that the Hamiltonian $H$ has $\mathcal{C}$ and $\mathcal{T}$ symmetry, respectively.
\end{definition}

\begin{remark}
    If the Hamiltonian $H$ has both $\mathcal{C}$ and $\mathcal{T}$ symmetries, it also has the chiral symmetry, i.e., there exists an unitary operator $\Gamma \propto CT$, which anti-commutes with $H$, $\{\Gamma, H\} = 0$.
\end{remark}

There are two possible realizations of $\mathcal{C}$ and $\mathcal{T}$ operators, such that
\begin{align}
    \mathcal{C}^2 = \pm 1
    \, , \qquad  
    \mathcal{T}^2 = \pm 1
    \, ,
\end{align}
from which we obtain the AZ tenfold way classification~\cite{Altland:1997zz}.
We provide the summary of the classification in Table~\ref{tab:classification}.
The left-most column shows the symmetry class $\mathscr{C}$:
We both use the Cartan notation and the classifying space notation.
There are two complex and eight real classes.
Then, we show the classifying space and the space of time-evolution operator $U_{\mathscr{C}} = e^{iH}$ for each symmetry class.
We observe that the classifying space of class $\mathscr{C}_p$ agrees with the space of $U_{\mathscr{C}_{p+1}}$, where $p \in \mathbb{Z}_2$ ($\mathscr{C} = C$) and $p \in \mathbb{Z}_8$ ($\mathscr{C} = R$).
The right-most column shows $\mathcal{C}$ and $\mathcal{T}$ symmetries of each class.

\begin{table}[t]
\begin{center}
    \begin{tabular}{ccccccc} \toprule
        \multicolumn{2}{c}{Symmetry class} $\mathscr{C}$ &
         {Classifying space} $S_{\mathscr{C}}$ & T-evolution operator $U_{\mathscr{C}}$ & $\mathcal{T}^2$ & $\mathcal{C}^2$ & $\chi$ \\[.3em] \toprule
         A & $C_0$ & $\mathrm{U}/\mathrm{U} \times\mathrm{U} $ & $\mathrm{U}$ & 0 & 0 & 0 \\
         AIII & $C_1$ & $\mathrm{U} $ & $\mathrm{U}/\mathrm{U} \times\mathrm{U} $ & 0 & 0 & 1 \\ \midrule
         AI & $R_0$ & $\mathrm{O}/\mathrm{O} \times\mathrm{O} $ & $\mathrm{U} /\mathrm{O} $ & $+1$ & 0 & 0 \\
         BDI & $R_1$ & $\mathrm{O} $ & $\mathrm{O}/\mathrm{O} \times\mathrm{O} $ & $+1$ & $+1$ & 1 \\
         D & $R_2$ & $\mathrm{O}/\mathrm{U} $ & $\mathrm{O} $ & 0 & $+1$ & 0 \\
         DIII & $R_3$ & $\mathrm{U}/\mathrm{Sp} $ & $\mathrm{O}/\mathrm{U} $ & $-1$ & $+1$ & 1 \\
         AII & $R_4$ & $\mathrm{Sp}/\mathrm{Sp} \times\mathrm{Sp} $ & $\mathrm{U}/\mathrm{Sp} $ & $-1$ & 0 & 0 \\
         CII & $R_5$ & $\mathrm{Sp} $ & $\mathrm{Sp}/\mathrm{Sp} \times\mathrm{Sp} $ & $-1$ & $-1$ & 1 \\
         C & $R_6$ & $\mathrm{Sp} /\mathrm{U} $ & $\mathrm{Sp} $ & 0 & $-1$ & 0 \\
         CI & $R_7$ & $\mathrm{U} /\mathrm{O} $ & $\mathrm{Sp} /\mathrm{U} $ & $+1$ & $-1$ & 1 \\\bottomrule
    \end{tabular}
\end{center}
    \caption{The AZ tenfold way classification of the classifying spaces and the associated time-evolution operators with respect to $\mathcal{T}$, $\mathcal{C}$, and chiral $(\chi)$ symmetries.}
    \label{tab:classification}
\end{table}

\section{Bulk extension}\label{sec:bulk_ext}

Utilizing both formalisms of Lagrangian and Hamiltonian, we introduce the process of the bulk extension, which gives rise to the band flattened Hamiltonian.
We start with a $d$-dimensional gapped Hamiltonian of class $\mathscr{C}$ denoted by $H$.
Applying the Lagrangian formalism, we then obtain a $(d+1)$-dimensional Dirac operator by adding the 0-direction, $D = \gamma_0 H + \gamma_0 \partial_0$.
We apply lattice discretization to deal with this direction.
\begin{definition}
 We define the shift operator in the 0-direction denoted by $\nabla_0$, such that $\nabla_0 \psi_{n_0} = \psi_{n_0+1}$, where $\psi_{n_0}$ is the field operator with the 0-direction coordinate $n_0 \in \{ 1, \ldots, N \}$ with $N$ the size of the 0-direction.
\end{definition}
We do not explicitly write $d$-dimensional dependence of the field $\psi$ for simplicity.
Then, we define the $(d+1)$-dimensional Wilson--Dirac operator as follows.
\begin{definition}\label{def:Wilson-Dirac_op}
    Let $H$ be a $d$-dimensional gapped Hamiltonian of class $\mathscr{C}$.
    Let $\gamma \equiv \gamma_0$ and let $a$ be the lattice spacing constant in the 0-direction.
    Denoting the projection operator given by $P_\pm = \frac{1}{2}(\id \pm \gamma)$, we define the $(d+1)$-dimensional Wilson--Dirac operator,
    \begin{align}
	D = \gamma H - \frac{1}{a} P_+ \nabla_0 - \frac{1}{a} P_- \nabla_0^\dag + \frac{1}{a}
	\, .
    \end{align}
\end{definition}
\begin{remark}
If we do not impose the Wilson term, we instead have
\begin{align}
	D = \gamma H - \frac{1}{2a} \left( \nabla_0 - \nabla_0^\dag \right)
	\, ,
\end{align}
which involves additional contributions of species doublers in the low-energy regime.
\end{remark}

We evaluate the functional determinant of the Wilson--Dirac operator with the following boundary conditions in the 0-direction,
\begin{align}
    \nabla_0 \psi_N = 
    \begin{cases}
        0 & (\text{open}) \\
        +\psi_1 & (\text{periodic}) \\
        -\psi_1 & (\text{anti-periodic}) \\
    \end{cases}
\end{align}
\begin{definition}\label{def:eff_Dirac_det}
    Denoting the functional determinant with the boundary condition by $\det D_{\text{bc}}$ (bc $\in \{\text{op (open)},\text{p (periodic)}, \text{ap (anti-periodic)}\}$), we define the effective Dirac determinant, 
    \begin{align}
        \det \widetilde{D}_{\text{op}} = \frac{\det D_{\text{op}}}{\det D_{\text{ap}}}
        \, , \qquad
        \det \widetilde{D}_{\text{p}} = \frac{\det D_{\text{p}}}{\det D_{\text{ap}}}    
        \, .
    \end{align}
\end{definition}
\noindent
The denominator contribution with the anti-boundary contribution is known to be the Pauli--Villars regulator.
Then, we have the following.
\begin{proposition}\label{lemma:eff_Dirac_bulk_lim}    
Taking the large scale limit $N a \to \infty$, and then the continuum limit $a \to 0$ in 0-direction, we have the $d$-dimensional effective Dirac operator given by
\begin{align}
    \overline{D}_{\text{op}} := \lim_{a \to 0} \lim_{Na \to \infty} \widetilde{D}_{\text{op}} = \frac{1}{2} (1 + V )
    \, , \qquad
    \overline{D}_{\text{p}} := \lim_{a \to 0} \lim_{Na \to \infty} \widetilde{D}_{\text{p}} = V
    \, ,
\end{align}
where we define
\begin{align}
    V = \gamma \operatorname{sgn} H = \gamma \frac{H}{\sqrt{H^2}}
    \, .
\end{align}
\end{proposition}
\noindent
Hence, from the effective Dirac operator with the periodic boundary condition, we obtain the band flattened Hamiltonian $\overline{H} = \operatorname{sgn}(H)$.
\begin{theorem}\label{thm:band_flatten_Ham}
We have
\begin{align}
    \overline{H} = \gamma \overline{D}_{\text{p}} 
    \, .
\end{align}
\end{theorem}
\begin{proof}
 It immediately follows from Proposition~\ref{lemma:eff_Dirac_bulk_lim}.
\end{proof}
\noindent
On the other hand, the effective Dirac operator with the open boundary condition provides the so-called overlap Dirac operator $\overline{D}_{\text{op}} = D_{\text{ov}}$~\cite{Neuberger:1997fp,Neuberger:1997bg,Neuberger:1998wv}.
We will discuss it in more detail in Sec.~\ref{sec:GW_relation}.

\begin{remark}\label{rmk:Euler_characteristic}
    Redefining the $V$-operator $V \to -V$, and changing the normalization, we have the determinant of $\overline{D}_{\text{op}}$ as follows,
    \begin{align}
        \det (1 - V) = \sum_{i=0}^{\operatorname{rk}V} (-1)^i \tr \wedge^i V
        \, ,
    \end{align}
    which is interpreted as an equivariant analogue of Euler characteristic.
    See also Remark~\ref{rmk:Witten_ind}.
    Moreover, the ratio of determinants used in Definition~\ref{def:eff_Dirac_det} also implies a K-theory formulation, which involves the difference of vector bundles associated with each boundary condition.
\end{remark}

The operator $V$ is unitary, $V^\dag = V^{-1}$ since $\gamma$ and $\operatorname{sgn} H$ are hermitian, and $\gamma^2 = (\operatorname{sgn} H)^2 = 1$.
In fact, it has been known that the non-hermitian point-gap Hamiltonian is topologically equivalent to the unitary operator~\cite{Kawabata:2018gjv}.
We remark that such a unitary operator is also discussed in the context of Floquet systems (see, e.g.,~\cite{Sun:2018PRL,Bessho:2020hrs}).

For each symmetry class based on the AZ tenfold way classification, we have the following.
\begin{proposition}\label{prop:V_classifying_space}
    For class $\mathscr{C}$ system, the unitary operator $V$, hence the effective Dirac operator $\overline{D}_{\text{p}}$ takes value in the corresponding classifying space $S_{\mathscr{C}}$ in the $d$-dimensional bulk limit.
\end{proposition}

The remaining part of this Section is devoted to a proof of Proposition~\ref{lemma:eff_Dirac_bulk_lim} and Proposition~\ref{prop:V_classifying_space} for each symmetry class $\mathscr{C}$.

\begin{corollary}\label{cor:symmetry_shift}
For class $\mathscr{C}_p$ system, the operator $H_V$ defined by $V = e^{iH_V} \in S_{\mathscr{C}_p}$ belongs to class $\mathscr{C}_{p+1}$.
In other words, the symmetry of $H_V$ agrees with that of the Hamiltonian of class $\mathscr{C}_p$ in the gapless limit.
\end{corollary}
\begin{proof}
    This follows from that the classifying space $S_{\mathscr{C}_p}$ agrees with the space of time-evolution operators of class $\mathscr{C}_{p+1}$ as shown in Table~\ref{tab:classification}. 
    In the gapless limit, the mass matrix $\gamma$ plays a role of the additional symmetry operator, which changes the symmetry class $\mathscr{C}_{p}$ to $\mathscr{C}_{p+1}$.
\end{proof}

\subsection{Wigner--Dyson class}

We first apply the bulk extension formalism to the Wigner--Dyson class (class A, AI, AII; threefold way).
We in particular discuss the class A case ($\mathbb{C}$-hermitian Hamiltonian with no symmetry). 
The class AI and AII cases can be discussed in parallel by replacing the $\mathbb{C}$-Hamiltonian with those for $\mathbb{R}$ and $\mathbb{H}$.

\subsubsection{Class A}\label{sec:extension_A}

The class A Hamiltonian is given by a $\mathbb{C}$-hermitian matrix with no additional symmetry.
\begin{definition}
We consider a $d$-dimensional gapped system of class A, which is described by the following size $k$ Hamiltonian, 
\begin{align}
 H = 
 \begin{pmatrix}
  \mathsf{A} & \mathsf{C} \\ \mathsf{C}^\dag & \widetilde{\mathsf{A}}
 \end{pmatrix}
 \in \mathsf{H}(k,\mathbb{C})
 \, ,
    \label{eq:classA_Hamiltonian}
\end{align}    
where $k = k_1 + k_2$ and
\begin{align}
    \mathsf{A} \in \mathsf{H}(k_1,\mathbb{C})
    \, , \qquad 
    \widetilde{\mathsf{A}} \in \mathsf{H}(k_2,\mathbb{C})
    \, , \qquad 
    \mathsf{C} \in \mathbb{C}^{k_1 \times k_2}
    \, .
\end{align}
\end{definition}
This block matrix structure is taken with respect to the mass matrix, 
\begin{align}
 \gamma \equiv \gamma_0 = 
 \begin{pmatrix}
  \id_{k_1} & 0 \\ 0 & -\id_{k_2}
 \end{pmatrix}
 \, ,
\end{align}
and hence we have the projection operators,
\begin{align}
 P_+ = \frac{1 + \gamma}{2} =
 \begin{pmatrix}
  \id_{k_1} & 0 \\ 0 & 0
 \end{pmatrix}
 \, , \qquad
 P_- = \frac{1 - \gamma}{2} =
 \begin{pmatrix}
  0 & 0 \\ 0 & \id_{k_2}
 \end{pmatrix}
 \, .
\end{align}
In this case, applying Definition~\ref{def:Wilson-Dirac_op}, the Wilson--Dirac operator is given by
\begin{align}
 a D = 
 \begin{pmatrix}
  A & C \\ - C^\dag & B
 \end{pmatrix}
 - P_+ \nabla_0 - P_- \nabla_0^\dag 
\end{align}
where 
\begin{align}
    A = \id_{k_1} + a \mathsf{A}
    \, , \qquad 
    B = \id_{k_2} - a \widetilde{\mathsf{A}}
    \, , \qquad 
    C = a \mathsf{C}
    \, .
\end{align}

The next step is to compute the determinant of size $Nk = N(k_1 + k_2)$ to consider the effective Dirac operator,
\begin{align}
     \det a D & = 
    \begin{vmatrix}
    A & C & 0 &&&&& Y & 0 \\
    -C^\dag & B & 0 & -\id_{k_2} &&&& 0 & 0 \\
    - \id_{k_1} & 0 & A & C & 0 &&&& \\
    & 0 & -C^\dag & B & 0 & -\id_{k_2} &&& \\
    && -\id_{k_1} & 0 & A & C & \ddots && \\
    &&&&\ddots&\ddots&\ddots&& \\
    0 & 0 &&&&&& A & C \\
    0 & X &&&&&&-C^\dag& B
    \end{vmatrix}
\end{align}
where we take $X$ and $Y$ depending on the boundary condition in the 0-direction,
\begin{align}
 X, Y =
 \begin{cases}
  0 & (\text{open}) \\
  - \id & (\text{periodic}) \\
  + \id & (\text{anti-periodic})
 \end{cases}
\end{align}
In order to write down the determinant, we define the following operator.
\begin{definition}\label{def:T-op}
We define the hermitian $T$-operator (transfer matrix) as follows,
\begin{align}
    T = 
    \begin{pmatrix}
        C B^{-1} C^\dag + A & C B^{-1} \\
        B^{-1} C^\dag & B^{-1}
    \end{pmatrix}
    \, .
\end{align}
\end{definition}
\begin{remark}
    The determinant of the $T$-operator is given by
\begin{align}
    \det T & = \det \qty( C B^{-1} C^\dag + A - C B^{-1} \cdot B \cdot B^{-1} C^\dag ) \det \qty(B^{-1})
    = \frac{\det A}{\det B}
    \, .
\end{align}
\end{remark}
\begin{lemma}\label{lemma:Dirac_det_classA}
The Wilson--Dirac operator determinant is given as follows,
\begin{align}
    \det a D & = (-1)^{n} \det A^N \det \qty( \begin{pmatrix}
    \id_{k_1} & \\ & -X
    \end{pmatrix} - T^{-N} 
    \begin{pmatrix}
    - Y & \\ & \id_{k_2}
    \end{pmatrix})
    \nonumber \\
    & =
    \begin{cases}
    \displaystyle
    (-1)^{n} \det A^N \det \frac{1}{2} \Big( 1 - T^{-N} + \qty(1 + T^{-N}) \gamma  \Big)  & (\text{open})
    \\[.5em] \displaystyle
    (-1)^{n} \det A^N \det \qty( 1 - T^{-N}) & (\text{periodic})
    \\[.5em] \displaystyle
    (-1)^{n} \det A^N \det \qty( 1 + T^{-N}) \gamma & (\text{anti-periodic})
    \end{cases}
    \label{eq:Dirac_det_classA}
\end{align}
where $n = (N-1)k_2^2 + Nk_2$.
\end{lemma}
A proof of Lemma~\ref{lemma:Dirac_det_classA} is given in Appendix~\ref{sec:Dirac_det_Proof}.\\

From this expression, we obtain the effective Dirac operator determinant (Definition~\ref{def:eff_Dirac_det}).
\begin{lemma}\label{lemma:eff_Dirac_det_classA}
Define the effective Hamiltonian $\mathcal{H}$ through the $T$-operator $T =: e^{a \mathcal{H}}$.
Then, the Wilson--Dirac operator determinant is given by
\begin{subequations}
\begin{align}
    \det \widetilde{D}_{\text{op}} & = \frac{\det D_{\text{op}}}{\det D_{\text{ap}}}
    = \det \frac{1}{2} \qty( 1 + \gamma \frac{1 - T^{-N}}{1 + T^{-N}} )
    = \det \frac{1}{2} \qty( 1 + \gamma \tanh \qty(\frac{N a}{2} \mathcal{H}) )
    \, , \\
    \det \widetilde{D}_{\text{p}} & = \frac{\det D_{\text{p}}}{\det D_{\text{ap}}}
    = \det \qty( \gamma \frac{1 - T^{-N}}{1 + T^{-N}} )
    = \det \qty( \gamma \tanh \qty(\frac{N a}{2} \mathcal{H}) )
    \, ,
\end{align}
from which we obtain the effective Dirac operator,
\begin{align}
    \widetilde{D}_{\text{op}} = \frac{1}{2} \qty( 1 + \gamma \tanh \qty(\frac{N a}{2} \mathcal{H}) )
    \, , \qquad
    \widetilde{D}_{\text{p}} = \gamma \tanh \qty(\frac{N a}{2} \mathcal{H}) 
    \, .
\end{align}
\end{subequations}    
\end{lemma}
\begin{proof}
    This follows from Lemma~\ref{lemma:Dirac_det_classA}.
\end{proof}

In order to prove Proposition~\ref{lemma:eff_Dirac_bulk_lim} in this case, we take the following limits.
\begin{enumerate}
    \item Large scale limit: $N a \to \infty$\\[.5em]
    In this limit, the tanh function behaves as
    \begin{align}
        \lim_{N a \to \infty} \tanh \qty(\frac{N a}{2} \mathcal{H}) = \operatorname{sgn} \mathcal{H} = \frac{\mathcal{H}}{\sqrt{\mathcal{H}^2}}
        \, ,
    \end{align}
    where the spectrum is given by $\operatorname{Spec}(\operatorname{sgn} \mathcal{H}) = \{\pm 1 \}$.
    
    \item Continuum limit: $a \to 0$\\[.5em]
    In this limit, we have the expansion,
    \begin{align}
        T = \id_k + a \mathcal{H} + O(a^2) 
        \, .
    \end{align}
    On the other hand, we also have
    \begin{align}
        T =
        \begin{pmatrix}
        a \mathsf{C} (\id_{k_1} - a \widetilde{\mathsf{A}})^{-1} a \mathsf{C}^\dag + \id_{k_1} + a {\mathsf{A}} & a \mathsf{C} (\id_{k_2} - a \widetilde{\mathsf{A}})^{-1} \\
        (\id_{k_2} - a \widetilde{\mathsf{A}})^{-1} a \mathsf{C}^\dag & (\id_{k_2} - a \widetilde{\mathsf{A}})^{-1}
    \end{pmatrix}
    = \id_k + a H + O (a^2)
    \, ,
    \end{align}
    from which we obtain
    \begin{align}
        \lim_{a \to 0} \mathcal{H} = H 
        \, .
    \end{align}

\end{enumerate}

\begin{proposition}
    Proposition~\ref{lemma:eff_Dirac_bulk_lim} holds for class A.
\end{proposition}
\begin{proof}
    Taking the large scale limit, and then the continuum limit, we have
    \begin{align}
        \lim_{a \to 0} \lim_{N a \to 0} \gamma \tanh \qty(\frac{N a}{2} \mathcal{H}) = \gamma \operatorname{sgn} H = V
        \, .
    \end{align}
    Then, it follows from Lemma~\ref{lemma:eff_Dirac_det_classA}.
\end{proof}

\begin{proposition}
    Proposition~\ref{prop:V_classifying_space} holds for class A.
\end{proposition}
\begin{proof}
    We first remark that the unitary operator $V$ obeys $\gamma V \gamma = \operatorname{sgn} (H) \gamma = V^\dag$.
    Hence, parametrizing $V = e^{X}$, we obtain $X^\dag = -X$ and also $\{ \gamma, X\} = 0$, which implies that $V$ takes a value in the complex Grassmannian, which becomes the classifying space of class A in the inductive limit,
\begin{align}
    V \in \bigcup_{k_1 + k_2 = k} \frac{\mathrm{U}(k)}{\mathrm{U}(k_1) \times \mathrm{U}(k_2)} \xrightarrow{k \to \infty} C_0 \, .
    \label{eq:classifying_sp_C0}
\end{align}
This large $k$ limit corresponds to the thermodynamic limit of the $d$-dimensional bulk system (bulk limit).
\end{proof}

The classifying space plays an important role to discuss the topological property of the system.
The zero-th homotopy group of $C_0$ is given by $\pi_0(C_0) = \mathbb{Z}$, and in general we have $\pi_{d}(C_0) = \mathbb{Z}$ for $d \in 2 \mathbb{Z}_{\ge 0}$.
In this case, we obtain the topological invariant of the $d$-dimensional gapped system $\nu \in \mathbb{Z}$.
\begin{proposition}
Identifying the case $\nu = 0$ as a topologically trivial case, we parametrize $k_1 = n - \nu$, $k_2 = n + \nu$.
Then, we obtain
\begin{align}
    \nu = - \frac{1}{2} \tr \operatorname{sgn} H = - \frac{1}{2} \eta(H) 
    \, ,
    \label{eq:top_num_classA}
\end{align}
where the eta invariant $\eta(H)$ is defined by
\begin{align}
    \eta(H) = \tr \operatorname{sgn} H
    \, .
\end{align}
\end{proposition}
This bulk topological invariant agrees with the index of the overlap Dirac operator, $\operatorname{ind} (D_{\text{ov}}) = \nu$ \cite{Luscher:1998pqa,Adams:1998eg}.
See Sec.~\ref{sec:GW_relation}.
We also remark that the eta invariant appears in the formalism of the domain-wall fermion from the Dirac determinant together with the Pauli--Villars regularization, which is analogous to the definition of $\widetilde{D}_{\text{p}}$ and $\overline{D}_{\text{p}}$.
See, e.g., a recent review~\cite{Fukaya:2021sea} for details.

\begin{remark}\label{rmk:Witten_ind}
    Writing the mass matrix $\gamma = (-1)^F$ and the $V$-operator $V = e^{iH_V}$, the overlap operator index (bulk topological invariant) is written in the form of the equivariant Witten index,
    \begin{align}
        \nu = - \frac{1}{2} \tr \left[ (-1)^F e^{iH_V} \right]
        \, .
    \end{align}
    See also Remark~\ref{rmk:Euler_characteristic}.
\end{remark}

\subsubsection*{Wilson--Dirac fermion in $d = 2$}

Let us demonstrate the bulk extension formalism for a gapped class A system in $d = 2$.
We consider the following Wilson--Dirac Hamiltonian,
\begin{align}
    H = 
    \begin{pmatrix}
        m + 2 - (\cos p_1 + \cos p_2) & \sin p_1 - i \sin p_2 \\
        \sin p_1 + i \sin p_2 & - m - 2 + (\cos p_1 + \cos p_2)
    \end{pmatrix}
    \, ,
    \label{eq:WD_Ham}
\end{align}
and the two-dimensional part of the corresponding Wilson--Dirac operator,
\begin{align}
    D = \gamma H =
    \begin{pmatrix}
        m + 2 - (\cos p_1 + \cos p_2) & \sin p_1 - i \sin p_2 \\
        - \sin p_1 - i \sin p_2 & m + 2 - (\cos p_1 + \cos p_2)
    \end{pmatrix}
    \, ,
    \label{eq:WD_op}
\end{align}
where the mass matrix is given by $\gamma = \sigma_3$.
We define the intermediate Hamiltonian and the $V$-operator of the finite size $N$ as follows,
\begin{align}
    H_N = \tanh \left( \frac{N}{2} \mathcal{H} \right)
    \, , \qquad 
    V_N = \gamma H_N
    \, .
\end{align}
The band spectra $E$ of the Hamiltonian $H_N$ with $m = -1$ are presented in Fig.~\ref{fig:H_spec}.
The case $N=0$ shows the spectrum of the Hamiltonian~\eqref{eq:WD_Ham} itself.
We see that the spectrum becomes flat as $N$ becomes large.

\begin{figure}
    \centering
    \begin{tikzpicture}
        \begin{scope}
        \node at (0,0) {\includegraphics[height=3.5cm]{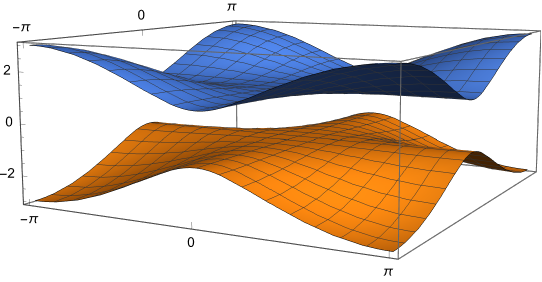}};
        \node at (-1.2,-1.6) {$p_1$};
        \node at (-1.8,1.9) {$p_2$};
        \node at (-3.7,.3) {$E$};
        \node at (-3,2) {\underline{$N=0$}};
        \end{scope}
        \begin{scope}[shift={(8,0)}]
        \node at (0,0) {\includegraphics[height=3.5cm]{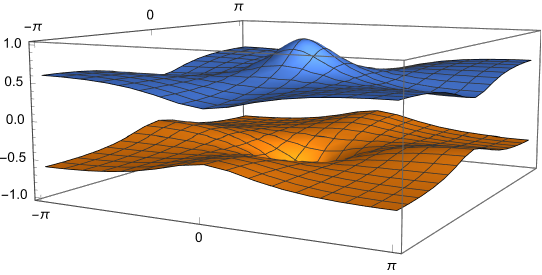}};
        \node at (-1.2,-1.6) {$p_1$};
        \node at (-1.8,1.9) {$p_2$};
        \node at (-3.7,.3) {$E$};
        \node at (-3,2) {\underline{$N=1$}};
        \end{scope}
        \begin{scope}[shift={(0,-4.5)}]
        \node at (0,0) {\includegraphics[height=3.5cm]{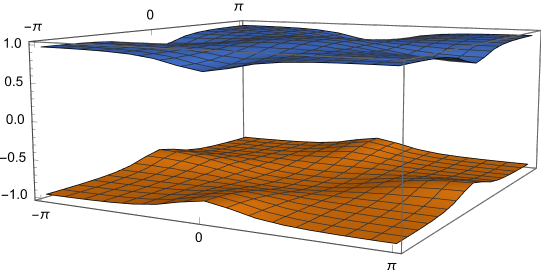}};
        \node at (-1.2,-1.6) {$p_1$};
        \node at (-1.8,1.9) {$p_2$};
        \node at (-3.7,.3) {$E$};
        \node at (-3,2) {\underline{$N=3$}};
        \end{scope}        
        \begin{scope}[shift={(8,-4.5)}]
        \node at (0,0) {\includegraphics[height=3.5cm]{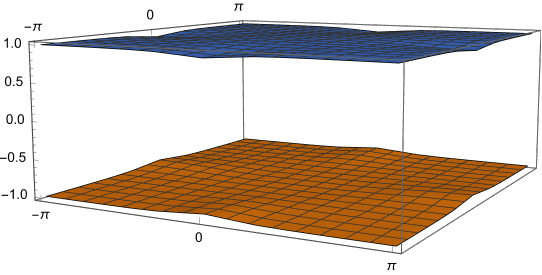}};
        \node at (-1.2,-1.6) {$p_1$};
        \node at (-1.8,1.9) {$p_2$};
        \node at (-3.7,.3) {$E$};
        \node at (-3,2) {\underline{$N=5$}};
        \end{scope}                
    \end{tikzpicture}
    \caption{The band spectra $E$ of the Wilson--Dirac Hamiltonian $H_N$ with $m = -1$. The case $N=0$ shows the spectrum of the Hamiltonian~\eqref{eq:WD_Ham}.}
    \label{fig:H_spec}
\end{figure}
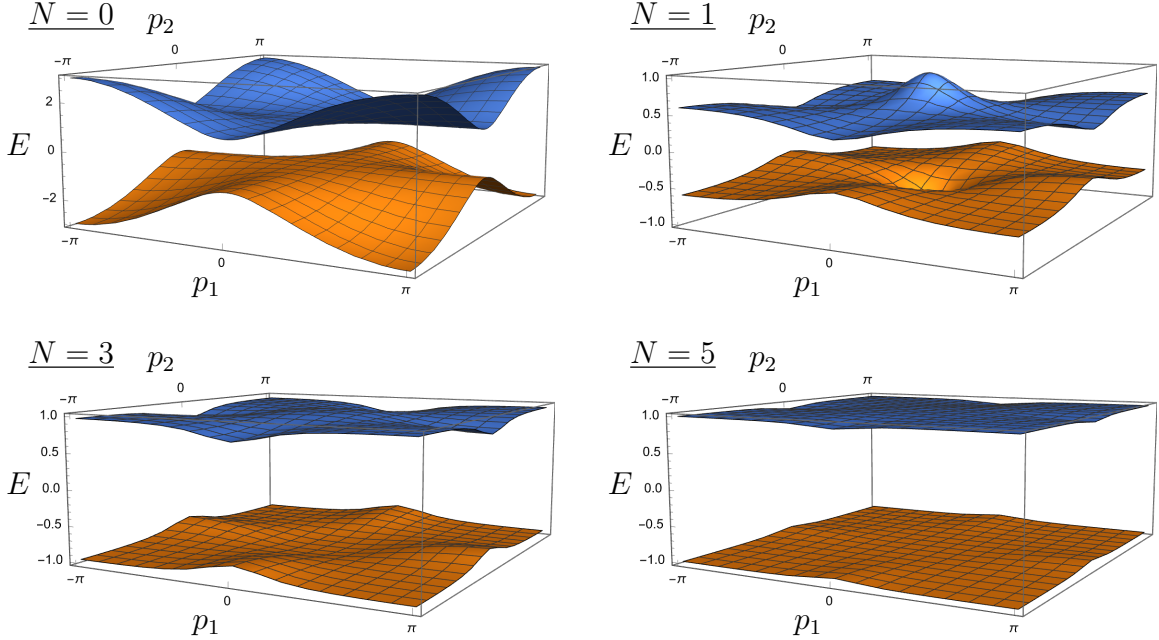

The complex spectrum of the Wilson--Dirac operator~\eqref{eq:WD_op} for $m = -1$ is given in Fig.~\ref{fig:WD}:
The horizontal and vertical axes are for the real part and imaginary part of the spectrum. 
We show the spectra of the $V$-operator at finite $N$ denoted by $V_N$ in Fig.~\ref{fig:V_spec}.
The spectrum approaches to a unit circle as $N$ becomes large.

\begin{figure}[t]
    \centering
    \begin{tikzpicture}
        \node at (0,0) {\includegraphics{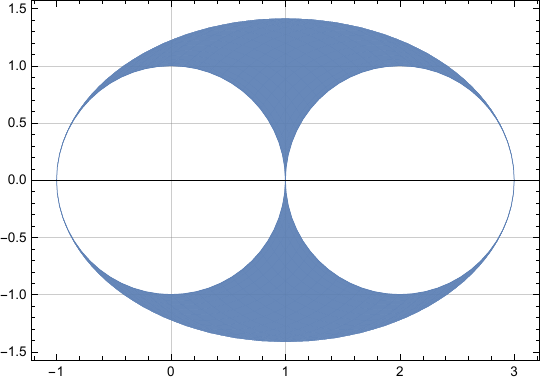}};
        \node at (.2,-3.5) {$\operatorname{Re}\lambda$};
        \node at (-5,.2) {$\operatorname{Im}\lambda$};
    \end{tikzpicture}
    \caption{The complex spectrum $\lambda$ of $d = 2$ Wilson--Dirac operator with $m = -1$.}
    \label{fig:WD}
\end{figure}

\begin{figure}[t]
    \centering
    \begin{tikzpicture}
    \begin{scope}
        \node at (0,0) {\includegraphics[width=5cm]{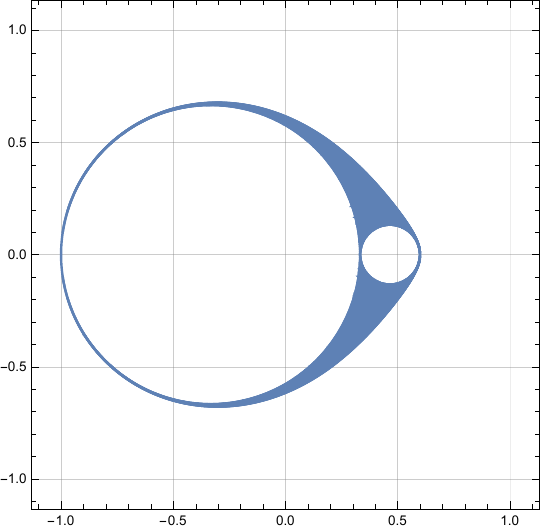}};
        \node at (1.7,1.9) {$N = 1$};
    \end{scope}
    \begin{scope}[shift={(5.5,0)}]
        \node at (0,0) {\includegraphics[width=5cm]{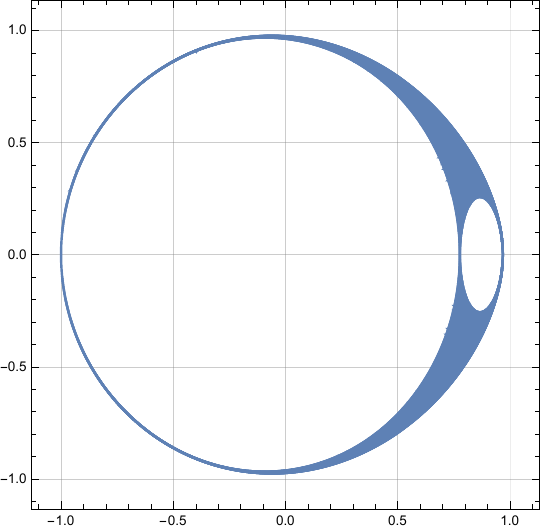}};
        \node at (1.7,1.9) {$N = 3$};
    \end{scope}
    \begin{scope}[shift={(11,0)}]
        \node at (0,0) {\includegraphics[width=5cm]{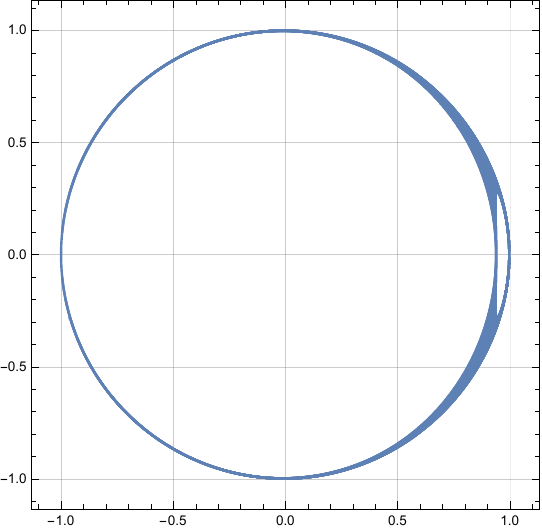}};
        \node at (1.7,1.9) {$N = 5$};
    \end{scope}    
    \begin{scope}[shift={(0,-5.5)}]
        \node at (0,0) {\includegraphics[width=5cm]{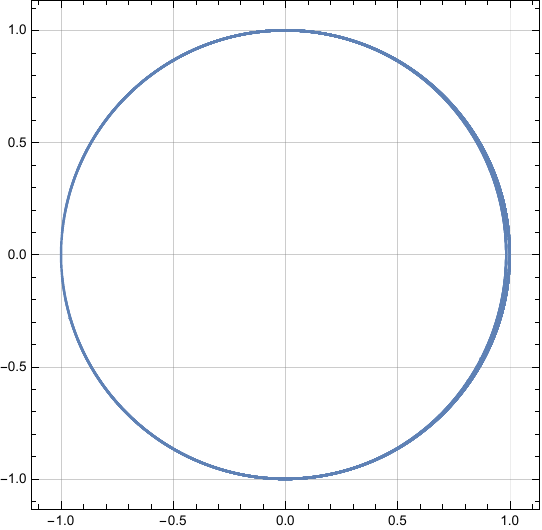}};
        \node at (1.7,1.9) {$N = 7$};
    \end{scope} 
    \begin{scope}[shift={(5.5,-5.5)}]
        \node at (0,0) {\includegraphics[width=5cm]{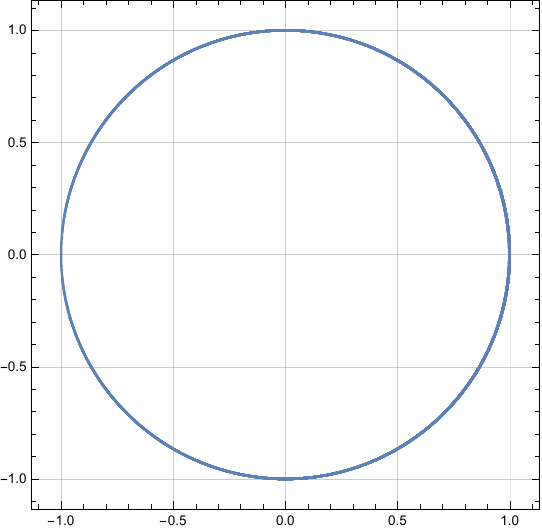}};
        \node at (1.7,1.9) {$N = 10$};
    \end{scope} 
    \begin{scope}[shift={(11,-5.5)}]
        \node at (0,0) {\includegraphics[width=5cm]{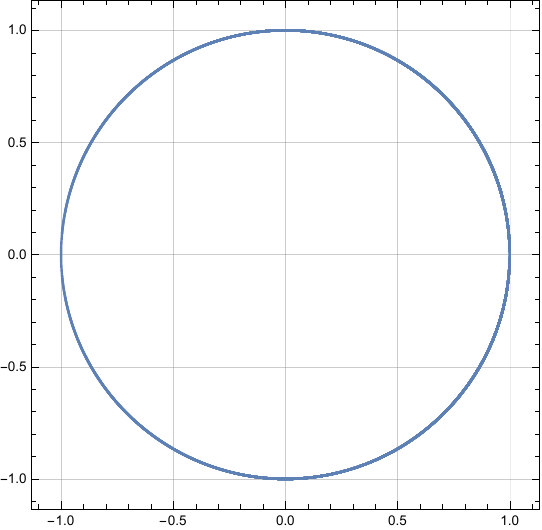}};
        \node at (1.7,1.9) {$N = 20$};
    \end{scope}     
    \end{tikzpicture}
    \caption{The complex spectra of $V_N$. The horizontal and vertical axes are for the real part and imaginary part of the spectrum.}
    \label{fig:V_spec}
\end{figure}

\subsubsection{Class AI, AII}

Let us consider the other Wigner--Dyson classes, class AI and AII.
In these cases, we can apply the same analysis after replacing the $\mathbb{C}$-hermitian Hamiltonian of class A~\eqref{eq:classA_Hamiltonian} with the $\mathbb{R}$-symmetric and the $\mathbb{H}$-self-dual matrices for class AI and AII, respectively,
\begin{subequations}
\begin{align}
    (\text{AI}) \quad &
    \mathsf{A} \in \mathsf{H}(k_1,\mathbb{R})
    \, , \qquad 
    \widetilde{\mathsf{A}} \in \mathsf{H}(k_2,\mathbb{R})
    \, , \qquad 
    \mathsf{C} \in \mathbb{R}^{k_1 \times k_2}
    \, , \\
    (\text{AII}) \quad &
    \mathsf{A} \in \mathsf{H}(k_1,\mathbb{H})
    \, , \qquad 
    \widetilde{\mathsf{A}} \in \mathsf{H}(k_2,\mathbb{H})
    \, , \qquad 
    \mathsf{C} \in \mathbb{H}^{k_1 \times k_2}
    \, .
\end{align}
\end{subequations}

\begin{proposition}
    Proposition~\ref{lemma:eff_Dirac_bulk_lim} and Proposition~\ref{prop:V_classifying_space} hold for class AI and AII.
\end{proposition}
\begin{proof}
    In the real and quaternion cases, we should replace the imaginary unit $i = \sqrt{-1}$ with the gamma matrix $\gamma$ obeying $\gamma^2 = -1$.
    Since we have not used it explicitly in Sec.~\ref{sec:extension_A}, we may apply the same argument to these cases as before.
    The $V$-operator takes a value in the real and quaternion Grassmannians, which become the classifying spaces of class AI and AII in the bulk limit,    
\begin{align}
    V \in 
    \begin{cases}
    \displaystyle
    \bigcup_{k_1 + k_2 = k} \frac{\mathrm{O}(k)}{\mathrm{O}(k_1) \times \mathrm{O}(k_2)} \quad \xrightarrow{k \to \infty} \ R_0 & (\text{AI}) \\
    \displaystyle
    \bigcup_{k_1 + k_2 = k} \frac{\mathrm{Sp}(k)}{\mathrm{Sp}(k_1) \times \mathrm{Sp}(k_2)} \ \xrightarrow{k \to \infty} \ R_4 & (\text{AII}) 
    \end{cases}
\end{align}
\end{proof}

For class AI and AII, we have $\pi_0(R_0) = \pi_0(R_4) = \mathbb{Z}$.
We can similarly obtain the bulk topological invariant~\eqref{eq:top_num_classA}, which agrees with the index of overlap Dirac operator.

\subsection{Chiral class}

Let us then consider the chiral class.
We focus on the complex case (class AIII) for the moment. 
The other classes (class BDI, CII) are discussed in the same way by replacing by $\mathbb{R}$ and $\mathbb{H}$ matrices.

\begin{definition}\label{def:classAIII_Hamiltonian}
We define the $d$-dimensional gapped class AIII Hamiltonian of size $2n$ by
\begin{align}
    H = 
    \begin{pmatrix}
        0 & \mathsf{C} \\ \mathsf{C}^\dag & 0
    \end{pmatrix}
    \, , \qquad \mathsf{C} \in \mathbb{C}^{n \times n} 
    \, .
\end{align}
\end{definition}

This Hamiltonian is obtained by taking $\mathsf{A}$, $\widetilde{\mathsf{A}} \to 0$ of Hamiltonian of class A~\eqref{eq:classA_Hamiltonian} with $k_1 = k_2 \equiv n$ ($k = 2n$).
This Hamiltonian possesses the chiral symmetry, $\{ \Gamma, H \} = 0$ with $\Gamma = \sigma_3 \otimes \id_n$, and the mass matrix is taken to be $\gamma \equiv \gamma_0 = \sigma_1 \otimes \id_n$.
Due to this chiral symmetry, all the non-zero eigenvalues make a pair, $\pm \lambda \in \operatorname{Spec}(H)$.
Since we assume that the Hamiltonian is gapped, the matrix size must be even, and the Hamiltonian takes a form as given in Definition~\ref{def:classAIII_Hamiltonian}.

In order to apply the same analysis as in Sec.~\ref{sec:extension_A} to the current case, we apply an orthogonal transformation.
We define an orthogonal matrix,
\begin{align}
    O = \frac{1}{\sqrt{2}}
    \begin{pmatrix}
        1 & 1 \\ 1 & -1
    \end{pmatrix} \otimes \id_n
    \, , \qquad 
    O^2 = \id_{2n} 
    \, ,
\end{align}
    which converts the gamma matrices,
\begin{align}
    \tilde{\Gamma} = O \Gamma O = \sigma_1 \otimes \id_n 
    \, , \qquad 
    \tilde{\gamma} = O \gamma O = \sigma_3 \otimes \id_n 
    \, ,
\end{align}
and the Hamiltonian,
\begin{align}
    \tilde{H} = O H O = \frac{1}{2}
    \begin{pmatrix}
        \mathsf{C} + \mathsf{C}^\dag & - \mathsf{C} + \mathsf{C}^\dag \\ \mathsf{C} - \mathsf{C}^\dag & -\mathsf{C} - \mathsf{C}^\dag 
    \end{pmatrix}
    \, .
\end{align}

Applying the bulk extension formalism to this case, we obtain the band flattened Hamiltonian $\overline{H}$ from the unitary operator $V = \tilde{\gamma} \operatorname{sgn}(\tilde{H})$ having the following properties.
\begin{lemma}\label{lemma:V-op_classAIII}
    The unitary operator $V = \tilde{\gamma} \operatorname{sgn}(\tilde{H})$ obeys
    \begin{align}
    V^\dag = V^{-1}
    \, , \qquad 
    \tilde{\gamma} V \tilde{\gamma} = V^\dag 
    \, , \qquad 
    \tilde{\Gamma} V \tilde{\Gamma} = V
    \, .
\end{align}
\end{lemma}
\begin{proof}
    The first two relations are straightforward.
    The third relation can be shown using $\{\tilde{\gamma},\tilde{\Gamma}\} = 0$ and $\{\tilde{\Gamma},\tilde{H}\} = 0$.
\end{proof}

\begin{proposition}
    Proposition~\ref{lemma:eff_Dirac_bulk_lim} and Proposition~\ref{prop:V_classifying_space} hold for the chiral classes AIII, BDI, and CII.
\end{proposition}
\begin{proof}
    Proposition~\ref{lemma:eff_Dirac_bulk_lim} can be shown in the same way as class A.
    For Proposition~\ref{prop:V_classifying_space}, we parametrize the $V$-operator as $V = e^{X}$ for class AIII.
    We can fix it from the relations in Lemma~\ref{lemma:V-op_classAIII} as follows,
\begin{align}
    X = 
    \begin{pmatrix}
        0 & Y \\ Y & 0
    \end{pmatrix}
    \, , \qquad 
    Y^\dag = -Y
    \, ,
\end{align}
which transforms under the unitary transformation, $X \to U X U^\dag$ with $U \in \mathrm{U}(n) \times \mathrm{U}(n) / \mathrm{U}(n) = \mathrm{U}(n)$.
In other words, $X \in \operatorname{Lie}(\mathrm{U}(n) \times \mathrm{U}(n) / \mathrm{U}(n)) = \mathfrak{u}(n)$.
Hence, the $V$-operator takes a value in the unitary group, which becomes the classifying space of class AIII in the bulk limit,
\begin{align}
    V \in \mathrm{U}(n) \xrightarrow{n \to \infty} C_1 \quad (\text{AIII})
    \, .
\end{align}
For the other chiral classes (class BDI and CII), we can show by replacing the $\mathbb{C}$-matrix with $\mathbb{R}$- and $\mathbb{H}$-matrices that the $V$-operator takes a value in the corresponding classifying space,
\begin{align}
    V \in 
    \begin{cases}
        \mathrm{O}(n) \ \xrightarrow{n \to \infty} R_1 & (\text{BDI}) \\
        \mathrm{Sp}(n) \xrightarrow{n \to \infty} R_5 & (\text{CII}) 
    \end{cases}
\end{align}
\end{proof}

We recall that $\pi_0(R_1) = \mathbb{Z}_2$ and $\pi_0(R_5) = 0$, and the mod-two bulk topological invariant of class (B)DI system denoted by $\nu$ is determined by the determinant of $V$ (see, e.g., \cite{Kitaev:2009mg}),
\begin{align}
    (-1)^{\nu} = \det V
    \, , 
    \label{eq:mod2index}
\end{align}
which would be identified with the mod-two index of the corresponding overlap Dirac operator $\operatorname{ind}(D_{\text{ov}}) \in \mathbb{Z}_2$.
See Sec.~\ref{sec:ind_thm}.

\subsection{BdG class}

There are four Bogoliubov-de Gennes (BdG) classes (class D, DIII, C, CI) described by the following Hamiltonian.
\begin{definition}
We define the $d$-dimensional gapped Hamiltonian of size $2n$,
\begin{align}
    H = 
    \begin{pmatrix}
        \mathsf{A} & \mathsf{C} \\ \mathsf{C}^\dag & - \mathsf{A}^{\text{T}}
    \end{pmatrix}
    \label{eq:BdG_Hamiltonian}
\end{align}
with
\begin{align}
    \mathsf{A} \in \mathsf{H}(n,\mathbb{C})
    \, , \qquad 
    \mathsf{C} \in \mathbb{C}^{n \times n}
    \, ,
\end{align}
which describes four BdG classes,
\begin{subequations}
\begin{align}
    \text{D :} \quad \mathsf{C}^{\text{T}} = - \mathsf{C}
    \, , \qquad
    \text{DIII :} \quad & \mathsf{C}^{\text{T}} = - \mathsf{C} 
    \, , \ \mathsf{A} = 0
    \, , \\
    \text{C :} \quad \mathsf{C}^{\text{T}} = + \mathsf{C}
    \, , \hspace{2.85em} 
    \text{CI :} \quad &  \mathsf{C}^{\text{T}} = + \mathsf{C} 
    \, , \ \mathsf{A} = 0
    \, .    
\end{align}
\end{subequations}
\end{definition}

\subsubsection{Class D}

We consider the gapped Hamiltonian of class D in the form of \eqref{eq:BdG_Hamiltonian} with the condition $\mathsf{C}^{\text{T}} = - \mathsf{C}$.
We define a unitary matrix,
\begin{align}
    U = \frac{1}{\sqrt{2}}
    \begin{pmatrix}
        1 & 1 \\ i & -i
    \end{pmatrix} \otimes \id_n
    \in \mathrm{U}(2n)
    \, ,
    \label{eq:U_basis}
\end{align}
which converts the mass matrix $\gamma = \sigma_3 \otimes \id_n$ to $\tilde{\gamma} = U \gamma U^\dag = \sigma_2 \otimes \id_n$, 
and the Hamiltonian,
\begin{align}
    \tilde{H} = U H U^\dag = i
    \begin{pmatrix}
       \alpha_I + \beta_I & - \alpha_R + \beta_R \\
       \alpha_R + \beta_R & \alpha_I - \beta_I
    \end{pmatrix}
    \, ,
\end{align}
where we denote $\mathsf{A} = \alpha_R + i \alpha_I$, $\mathsf{C} = \beta_R + i \beta_I$ with $\alpha_R, \alpha_I, \beta_R, \beta_I \in \mathbb{R}^{n \times n}$.
We remark that
$\alpha_R^{\text{T}} = \alpha_R$,
$\alpha_I^{\text{T}} = -\alpha_I$,
$\beta_R^{\text{T}} = -\beta_R$,
$\beta_I^{\text{T}} = -\beta_I$,
and hence $M := -i\tilde{H} \in \widetilde{\mathsf{H}}(2n,\mathbb{R}) = \mathfrak{o}(2n)$.

\begin{proposition}
    Proposition~\ref{lemma:eff_Dirac_bulk_lim} and Proposition~\ref{prop:V_classifying_space} hold for class D.
\end{proposition}
\begin{proof}
The proof of Proposition~\ref{lemma:eff_Dirac_bulk_lim} is the same as before.
Applying the bulk extension formalism for class D, we obtain the flat band Hamiltonian from $V = U (\gamma \operatorname{sgn} H) U^\dag = \tilde{\gamma} \operatorname{sgn} \tilde{H}$, which is an orthogonal matrix $V^{\text{T}} = V^{-1}$ with the property $\tilde{\gamma} V \tilde{\gamma} = V^{-1}$.
Parametrizing $V = e^X$, the matrix $X$ is given in the form of
\begin{align}
    X =
    \begin{pmatrix}
        \alpha & \beta \\ \beta & - \alpha
    \end{pmatrix}
\end{align}
where $\alpha^{\text{T}} = - \alpha$, $\beta^{\text{T}} = - \beta$.
On the other hand, a generic $\mathbb{R}$-skew-symmetric matrix $Z \in \mathfrak{o}(2n)$ has a decomposition,
\begin{align}
    Z = 
    \begin{pmatrix}
        \alpha + \delta & \beta + \beta' \\ \beta - \beta' & - \alpha + \delta
    \end{pmatrix}
    = 
    \begin{pmatrix}
        \alpha & \beta \\ \beta & - \alpha
    \end{pmatrix}
    + 
    \begin{pmatrix}
        \delta & \beta' \\ - \beta' & \delta 
    \end{pmatrix}
    \, ,
\end{align}
where $\delta^{\text{T}} = - \delta$, $\beta^{\prime \text{T}} = \beta'$.
Writing the second matrix as $\id_2 \otimes \delta + i \sigma_2 \otimes \beta'$, it is isomorphic to an anti-hermitian matrix, which is an element of the Lie algebra $\mathfrak{u}(n)$.
Hence, we have $X \in \operatorname{Lie}(\mathrm{O}(2n)/\mathrm{U}(n))$, which shows that the $V$-operator takes a value in the classifying space of class D in the bulk limit,
\begin{align}
    V \in \frac{\mathrm{O}(2n)}{\mathrm{U}(n)} \ \xrightarrow{n \to \infty} \ R_2 
    \quad (\text{class D})
    \, .
\end{align}
\end{proof}
\begin{remark}
    Recalling $\pi_0(R_2) = \mathbb{Z}_2$, the mod-two topological invariant is given in the same way as class BDI~\eqref{eq:mod2index}.
\end{remark}

\subsubsection{Class C}

Let us discuss the class C system described by the BdG Hamiltonian~\eqref{eq:BdG_Hamiltonian} with $\mathsf{C}^{\text{T}} = \mathsf{C}$.

We apply the same basis change matrix \eqref{eq:U_basis}, and define an $\mathbb{H}$-matrix of size $n$ as follows,
\begin{align}
    \check{H}_{jk} = i
    \begin{pmatrix}
     \alpha_{I,jk} - i \beta_{I,jk} & - \alpha_{R,jk} + i \beta_{R,jk} \\
     \alpha_{R,jk} + i \beta_{R,jk} & \alpha_{I,jk} + i \beta_{I,jk}
    \end{pmatrix}
    \in i \mathbb{H}
    \, , \quad 
    j, k = 1,\ldots, n
    \, ,
\end{align}
where
$\alpha_R^{\text{T}} = \alpha_R$,
$\alpha_I^{\text{T}} = -\alpha_I$,
$\beta_R^{\text{T}} = \beta_R$,
$\beta_I^{\text{T}} = \beta_I$.
We then define $M := -i\check{H} \in \mathbb{H}^{n \times n}$.
In fact, $M \in \mathfrak{sp}(n)$.

\begin{proposition}
    Proposition~\ref{lemma:eff_Dirac_bulk_lim} and Proposition~\ref{prop:V_classifying_space} hold for class C.
\end{proposition}
\begin{proof}
    The proof of Proposition~\ref{lemma:eff_Dirac_bulk_lim} is the same as before.
    In this case, we have the $\mathbb{H}$-valued $V$-operator $V = \tilde{\gamma} \operatorname{sgn} \check{H}$, and hence the flat band Hamiltonian is given by $\overline{H} = \tilde{\gamma} V$.
    Parametrizing $V = e^X$, each element of $X$ is given by
    \begin{align}
        X_{jk} = i
        \begin{pmatrix}
            \alpha_{jk} & \beta_{ji} \\ \beta_{jk} & - \alpha_{jk}
        \end{pmatrix}
        \in \mathbb{H}
    \end{align}
    where $\alpha = (\alpha_{jk})_{j,k = 1,\ldots,n}$ and $\beta = (\beta_{jk})_{j,k = 1,\ldots,n}$ are $\mathbb{R}$-symmetric matrices. 
    Compared with a generic $\mathfrak{sp}(n)$ element
    \begin{align}
        Z_{jk} = 
        \begin{pmatrix}
            \delta_{jk} + i \alpha_{jk} & \beta'_{jk} + i \beta_{jk} \\ 
            - \beta'_{jk} + i \beta_{jk} & \delta_{jk} - i \alpha_{jk} 
        \end{pmatrix}
        = 
        i
        \begin{pmatrix}
            \alpha_{jk} & \beta_{jk} \\ 
            \beta_{jk} & - \alpha_{jk} 
        \end{pmatrix}
        + 
        \begin{pmatrix}
            \delta_{jk} & \beta'_{jk} \\ 
            - \beta'_{jk} & \delta_{jk}
        \end{pmatrix}
        \, ,
    \end{align}
    with $\delta^{\text{T}} = - \delta$ and $\beta^{\prime\text{T}} = \beta'$, we have $X \in \operatorname{Lie}(\mathrm{Sp}(n)/\mathrm{U}(n))$.
    Hence, the $V$-operator takes a value in the classifying space of class C in the bulk limit,
    \begin{align}
    V \in \frac{\mathrm{Sp}(n)}{\mathrm{U}(n)} \ \xrightarrow{n \to \infty} \ R_6 
    \quad (\text{class C})
    \, .
    \end{align}
\end{proof}

\subsubsection{Class DIII}

For class DIII, the Hamiltonian is given by
\begin{align}
    H = 
    \begin{pmatrix}
        0 & \mathsf{C} \\ \mathsf{C}^\dag & 0
    \end{pmatrix}
    \, , \qquad 
    \mathsf{C}^{\text{T}} = - \mathsf{C}
    \, ,
    \label{eq:Ham_DIII}
\end{align}
which has $\mathcal{C}$ and $\mathcal{T}$ symmetries, such that $\mathcal{C}^2 = +1$, $\mathcal{T}^2 = -1$ (See Table~\ref{tab:classification}).
Provided that the Hamiltonian has a gap, we consider the matrix $\mathsf{C}$ of size $2n$, $\mathsf{C} \in \mathbb{C}^{2n \times 2n}$.
Hence, in this case, we may apply the following form of the symmetry matrices,
\begin{subequations}
\begin{align}
    C = \sigma_1 \otimes \sigma_3 \otimes \id_{n}
    \, , \qquad &
    T = i\sigma_2 \otimes \sigma_3 \otimes \id_{n}
    \, , \\ 
    \Gamma = \sigma_3 \otimes \id_2 \otimes \id_n
    \, , \qquad &
    \gamma = \sigma_1 \otimes \sigma_1 \otimes \id_n
    \, .
\end{align}
\end{subequations}
\begin{lemma}\label{lemma:CTchi_DIII}
    We have the $V$-operator, $V = \gamma \operatorname{sgn}H$, which behaves as follows,
    \begin{align}
        C V C^{-1} = T V T^{-1} = V^*
        \, , \quad 
        \Gamma V \Gamma = V
        \, , \quad 
        \gamma V \gamma = V^\dag
        \, .
    \end{align}
\end{lemma}
\begin{proof}
It follows from Definition~\ref{def:CT_Ham_action} together with the relations $\{ C, \gamma \} = 0$, $[T, \gamma] = 0$, $[\Gamma, \gamma] = 0$.
\end{proof}
\begin{proposition}
    Proposition~\ref{lemma:eff_Dirac_bulk_lim} and Proposition~\ref{prop:V_classifying_space} hold for class DIII.
\end{proposition}
\begin{proof}
    The proof of Proposition~\ref{lemma:eff_Dirac_bulk_lim} is the same as before.
    We parametrize $V = e^X$.
    From the behavior under the $\Gamma$-matrix shown in Lemma~\ref{lemma:CTchi_DIII}, we have
    $X = i \begin{pmatrix}
    Y & 0 \\ 0 & \tilde{Y}
    \end{pmatrix}$
    where $Y$, $\tilde{Y} \in {\mathsf{H}}(2n,\mathbb{C})$.
    We denote $\Sigma_i = \sigma_i \otimes \id_n$ for $i = 1,2,3$.
    Then, from the other relations, we have 
    \begin{align}
        \begin{cases}
        \Sigma_3 Y \Sigma_3 = -\tilde{Y}^* \\
        \Sigma_3 \tilde{Y} \Sigma_3 = -{Y}^* 
        \end{cases}
        \, , \qquad
        \begin{cases}
        \Sigma_1 Y \Sigma_1 = -\tilde{Y} \\
        \Sigma_1 \tilde{Y} \Sigma_1 = -{Y} 
        \end{cases}        
        \, .
    \end{align}
    Hence, we have $\Sigma_2 Y \Sigma_2 = Y^*$ and $\Sigma_2 \tilde{Y} \Sigma_2 = \tilde{Y}^*$, from which we deduce that they are isomorphic to $\mathbb{H}$-self-conjugate matrices.
    Recalling $\mathsf{H}(n,\mathbb{H}) = \operatorname{Lie}(\mathrm{U}(2n)/\mathrm{Sp}(n))$, the $V$-operator takes a value in the classifying space of class DIII in the bulk limit,
    \begin{align}
    V \in \frac{\mathrm{U}(2n)}{\mathrm{Sp}(n)} \ \xrightarrow{n \to \infty} \ R_3
    \, .
    \end{align}
\end{proof}

\subsubsection{Class CI}

For class CI, we again have the Hamiltonian in the form \eqref{eq:Ham_DIII} with a symmetric $\mathsf{C}$, $\mathsf{C}^{\text{T}} = \mathsf{C}$.
It has $\mathcal{C}$ and $\mathcal{T}$ symmetries, such that $\mathcal{C}^2 = -1$, $\mathcal{T}^2 = +1$ (See Table~\ref{tab:classification}), and we consider the matrix $\mathsf{C}$ of size $2n$, $\mathsf{C} \in \mathbb{C}^{2n \times 2n}$.
In this case, we may apply the following form of the symmetry matrices,
\begin{subequations}\label{eq:symmetry_matrices_CI}
\begin{align}
    C = i\sigma_2 \otimes \sigma_1 \otimes \id_{n}
    \, , \qquad &
    T = \sigma_1 \otimes \sigma_1 \otimes \id_{n}
    \, , \\ 
    \Gamma = \sigma_3 \otimes \id_2 \otimes \id_n
    \, , \qquad &
    \gamma = \sigma_2 \otimes \sigma_2 \otimes \id_n
    \, .
\end{align}
\end{subequations}
\begin{proposition}
    Proposition~\ref{lemma:eff_Dirac_bulk_lim} and Proposition~\ref{prop:V_classifying_space} hold for class CI.
\end{proposition}
\begin{proof}
    The proof of Proposition~\ref{lemma:eff_Dirac_bulk_lim} is the same as before.
    Having $V = \gamma \operatorname{sgn} H$, we have the same relations as shown in Lemma~\ref{lemma:CTchi_DIII} with the symmetry matrices~\eqref{eq:symmetry_matrices_CI}.
    As in the case of class DIII, under the parametrization $V = e^X$, we have
    $X = i \begin{pmatrix}
    Y & 0 \\ 0 & \tilde{Y}
    \end{pmatrix}$
    where $Y$, $\tilde{Y} \in {\mathsf{H}}(2n,\mathbb{C})$.
    From the other relations, we have 
    \begin{align}
        \begin{cases}
        \Sigma_1 Y \Sigma_1 = -\tilde{Y}^* \\
        \Sigma_1 \tilde{Y} \Sigma_1 = -{Y}^* 
        \end{cases}
        \, , \qquad
        \begin{cases}
        \Sigma_2 Y \Sigma_2 = -\tilde{Y} \\
        \Sigma_2 \tilde{Y} \Sigma_2 = -{Y} 
        \end{cases}        
        \, ,
    \end{align}
    from which we deduce that $Y, \tilde{Y} \in \mathsf{H}(2n,\mathbb{R})$.
    Recalling $\mathsf{H}(n,\mathbb{H}) = \operatorname{Lie}(\mathrm{U}(n)/\mathrm{O}(n))$, the $V$-operator takes a value in the classifying space of class CI in the bulk limit,
    \begin{align}
        V \in \frac{\mathrm{U}(2n)}{\mathrm{O}(2n)} \ \xrightarrow{n \to \infty} \ R_7
    \, .
    \end{align}
\end{proof}

\section{Overlap Dirac operator}\label{sec:GW_relation}

In this Section, we discuss the symmetry of the overlap Dirac operator of class $\mathscr{C}$,
\begin{align}
    D \equiv D_{\text{ov}} = \frac{1}{a} (1 + V)
    \, , \qquad 
    V = \gamma \operatorname{sgn} (H) \in S_{\mathscr{C}} 
    \, ,
    \label{eq:Dov_GWrel}
\end{align}
where we change the normalization of the operator for simplicity: We denote the lattice spacing parameter by $a$ with mass dimension $[a] = -1$.

\subsection{Ginsparg--Wilson relation}

First of all, it is clear from the unitarity of the $V$-operator, $V^\dag = V^{-1}$, that the overlap operator obeys the following relation, that we call Ginsparg--Wilson (GW) relation.
\begin{proposition}[Bietenholz--Nishimura~\cite{Bietenholz:2000ca}]
The overlap Dirac operator obeys Ginsparg--Wilson relation,
\begin{align}
    D + D^\dag = a D^\dag D = a D D^\dag
    \, .
    \label{eq:GW_rel_antihermite}
\end{align}
\end{proposition}
\begin{remark}
    GW relation is originally formulated as ``a remnant of chiral symmetry''~\cite{Ginsparg:1981bj}, and hence the relation shown in~\eqref{eq:GW_rel_chiral} is usually called GW relation.
\end{remark}

We remark that the RHS of \eqref{eq:GW_rel_antihermite} is suppressed in the continuum limit $a \to 0$, from which we deduce a simplified relation, $D + D^\dag = 0$.
Namely, $D$ becomes anti-hermitian, $D^\dag = - D$ in this limit, which is a generic property of gapless Dirac operators.
From this point of view, GW relation~\eqref{eq:GW_rel_antihermite} is interpreted as a non-linear deformation of the anti-hermiticity of gapless Dirac operator.

\subsubsection{Chiral symmetry}

We then discuss GW relation of the overlap operator with additional symmetries.
We use the parametrization $V = e^X$ again.
For the class with the chiral symmetry in the gapless limit (e.g., class A), we have $\{ \gamma, X \} = 0$, which gives rise to the $\gamma$-hermiticity, $\gamma D \gamma = D^\dag$.
Hence, we may rewrite GW relation \eqref{eq:GW_rel_antihermite} as follows.
\begin{proposition}[Neuberger~\cite{Neuberger:1998wv}]
The overlap operator for the class having the chiral symmetry in the gapless limit obeys Ginsparg--Wilson relation,
\begin{align}
    \gamma D + D \gamma = a D \gamma D 
    \, .
    \label{eq:GW_rel_chiral}
\end{align}
\end{proposition}
\noindent
This was shown originally for class A.
As discussed before, this is interpreted as a non-linear deformation of the chiral symmetry, which reproduces $\{\gamma, D\} = 0$ in the limit $a \to 0$.

From GW relation, we can discuss a non-linear deformation of chiral transformation.
We may rewrite the relation \eqref{eq:GW_rel_chiral} as follows, 
\begin{align}
    \gamma D + D \hat{\gamma} = 0
    \, , \qquad
    \hat{\gamma} = \gamma (1 - a D) = \gamma V
    \, .
\end{align}
Then, the Dirac Lagrangian~\eqref{eq:HvsD} is invariant under the following transformation,
\begin{align}
    \psi \longrightarrow \hat{\gamma} \psi 
    \, , \qquad 
    \bar{\psi} \longrightarrow \bar{\psi} \gamma
    \, ,
\end{align}
which, on the other hand, gives rise to a non-trivial contribution to the Jacobian providing the chiral anomaly~\cite{Luscher:1998pqa,Fujikawa:1998if,Suzuki:1998yz}.
We remark that this is not a unique way to write down the transformation:
In general, we may rewrite GW relation \eqref{eq:GW_rel_chiral} as $(1 - a b D) \gamma D + D\gamma(1 - a b' D) = 0$ where $b+b'=1$.

\subsubsection{$\mathcal{C}$ and $\mathcal{T}$ symmetries}\label{label:GW_rel_CT}

For the system with $\mathcal{C}$, $\mathcal{T}$ symmetry, we have $\mathcal{C}$, $\mathcal{T}$ analog of GW relation as follows.
\begin{theorem}\label{thm:CT-GW_relation}
For the class with $\mathcal{C}$, $\mathcal{T}$ symmetry in the gapless limit, the overlap Dirac operator obeys  $\mathcal{C}$ and $\mathcal{T}$ analog of Ginsparg--Wilson relation,
\begin{align}
    C D + D^{\text{T}} C = a D^{\text{T}} C D
    \, , \qquad 
    T D + D^* T = a D^* T D
    \, .    
    \label{eq:GW_rel_CT}
\end{align}
\end{theorem}
\begin{proof}
    We parametrize $V = e^{iH_V}$ with $H_V^\dag = H_V$.
    For the class with $\mathcal{C}$ symmetry in the gapless limit, we have $C H_V C^{-1} = - H_V^*$, which gives rise to $C V C^{-1} = V^{*}$. 
    Noticing $D^\dag = (C D C^{-1})^{\text{T}}$, we obtain GW relation with respect to $\mathcal{C}$ symmetry,
\begin{align}
    C D + D^{\text{T}} C = a D^{\text{T}} C D
    \, .
    \label{eq:GW_rel_C}
\end{align}
    For the class with $\mathcal{T}$ symmetry in the gapless limit, we instead have $T V T^{-1} = V^{\text{T}}$ and $D^\dag = (T D T^{-1})^*$, from which we obtain the corresponding GW relation,
\begin{align}
    T D + D^{*} T = a D^{*} T D
    \, .
    \label{eq:GW_rel_T}
\end{align}
\end{proof}
\noindent
They are again interpreted as a non-linear deformation of $\mathcal{C}$ and $\mathcal{T}$ symmetries of the gapless Dirac operator. 

Let us discuss the corresponding non-linear transformations.
We may rewrite GW relations \eqref{eq:GW_rel_CT} as follows,
\begin{align}
    CD + D^{\text{T}} \hat{C} = 0
    \, , \quad 
    TD + D^{*} \hat{T} = 0
    \, , \qquad 
    \hat{C} = CV \, , \ 
    \hat{T} = TV \, .
\end{align}
The corresponding non-linear $\mathcal{C}$ and $\mathcal{T}$ transformations are given by
\begin{align}
    \mathcal{C} \ : \ 
    \psi \longrightarrow \hat{C} \bar{\psi}^{\text{T}}
    \, , \quad 
    \bar{\psi} \longrightarrow \psi^{\text{T}} C^{-1}
    \, , \qquad 
    \mathcal{T} \ : \ 
    \psi \longrightarrow 
    \hat{T} \psi 
    \, , \quad 
    \bar{\psi} \longrightarrow \bar{\psi} T^{-1}    
    \, .
\end{align}
Hence, under these transformations, the fermion path integral measure behaves as
\begin{align}
    \mathrm{d}\psi \mathrm{d}\bar{\psi} \longrightarrow \left( \det V \right)^{-1} \mathrm{d}\psi \mathrm{d}\bar{\psi}
    \, .
\end{align}
This Jacobian factor is related to the anomalous behavior of Majorana(--Weyl) fermion (hence, $\mathcal{C}$ transformation)~\cite{Huet:1996pw,Narayanan:1996mr,Inagaki:2004ar,Suzuki:2004ht,Hayakawa:2006fd}, and of the $\mathcal{T}$-invariant system~\cite{Fukui:2009pc,Ringel:2012fm}.
We also remark that a similar discussion is applied for the parity anomaly~\cite{Bietenholz:2000ca}.
These arguments are consistent with that the mod-two bulk topological invariant is given by the sign of $\det V$ as shown in \eqref{eq:mod2index}.

\subsection{Index theorem}\label{sec:ind_thm}

It has been known that the overlap formalism provides a concise way to understand the index theorem.
As mentioned in Proposition~\ref{prop:ind_thm_Z}, the $\mathbb{Z}$-valued index coincides with the bulk topological invariant of the corresponding system.
We have the following result for the mod-two index.

\begin{theorem}\label{thm:Z2_index}
The mod-two index of overlap Dirac operator, $\nu = \operatorname{ind}(D) = \dim \operatorname{ker}(D)$, is given by
\begin{align}
    (-1)^{\nu} = \det V \, .
\end{align}
\end{theorem}
We may have a non-trivial mod-two index when $\pi_d(R_p) = \mathbb{Z}_2$.
For $d = 0$, we have $\pi_0(R_1) = \pi_0(R_2) = \mathbb{Z}_2$ corresponding to class BDI and class D.
In fact, all the cases in $d > 0$ are reduced to these two classes via the dimensional reduction.
Hence, we focus on the case $d = 0$ to prove this Theorem.
We first consider a simplified situation.
\begin{lemma}\label{lemma:O2_index}
    Let $V \in \mathrm{O}(2)$ and $D = 1 + V$.
    Then, the mod-two index $\nu = \dim \operatorname{ker}(D)$ is given by
    \begin{align}
        (-1)^\nu = \det V
        \, .
    \end{align}
\end{lemma}
\begin{proof}
    We consider the following two elements of $\mathrm{O}(2)$,
    \begin{align}
        V_+ =
    \begin{pmatrix}
        \cos \lambda & \sin \lambda \\ - \sin \lambda & \cos \lambda
    \end{pmatrix}
        \, , \qquad 
        V_- =
    \begin{pmatrix}
        \cos \lambda & \sin \lambda \\ \sin \lambda & - \cos \lambda
    \end{pmatrix}
        \, ,
    \end{align}
    with the determinant $\det V_\pm = \pm 1$.
    Then, we have
    \begin{align}
        \dim \operatorname{ker} (1 + V_+) =
        \begin{cases}
            0 & (\lambda \neq \pi) \\ 2 & (\lambda = \pi)
        \end{cases}
        \, , \quad 
        \dim \operatorname{ker} (1 + V_-) = 1
        \, .
    \end{align}
    Hence, the mod-two index of $D$ depends only on the sign of $\det V$.
\end{proof}
Then, we apply this result to prove Theorem~\ref{thm:Z2_index}.
\begin{proof}[Proof of Theorem~\ref{thm:Z2_index}]
    We consider the case $V \in \mathrm{O}(n)$ for the moment.
    Let $v_i \in \mathrm{O}(2)$ ($i = 1,\ldots,m$, $m \le n/2$) and $\sigma_j \in \{ \pm 1 \} = \mathrm{O}(1)$ ($j = 2m+1,\ldots,n$).
    Then, there exists an orthogonal matrix $O$, such that
    \begin{align}
        O V O^{\text{T}} = 
        \begin{pmatrix}
            v_1 &&&&& \\
            & \ddots &&& 0 & \\
            && v_m &&& \\
            &&&\sigma_{2m+1}&& \\
            & 0 &&& \ddots & \\
            &&&&& \sigma_{n}
        \end{pmatrix}
        \, .
    \end{align}
    Hence, in this basis, we can apply Lemma~\ref{lemma:O2_index} for each block to obtain the index of $D_{\text{ov}}$, $\nu = \operatorname{ind}(D) = \dim \operatorname{ker}(D)$ as follows,
    \begin{align}
        (-1)^{\nu} = \det V \, .
    \end{align}
    For the case $V \in \mathrm{O}(2n)/\mathrm{U}(n)$, we may apply the same argument as in the case $V \in \mathrm{O}(2n)$.
    Taking the inductive limit $n \to \infty$, the $V$-operator takes a value in the corresponding classifying space.
\end{proof}

\appendix

\section{Proof of Lemma~\ref{lemma:Dirac_det_classA}}\label{sec:Dirac_det_Proof}

We follow the approach discussed in \cite{Neuberger:1998wv,Kikukawa:1998pd}.
We first define the permutation matrix,
\begin{align}
    \mathbf{P} = 
    \left(
    \begin{array}{ccc}\\[-1.3em]
    0 & \multicolumn{2}{c}{\id_{k_1}} \\[.5em]
    \id_{k_2} & \multicolumn{2}{c}{0}
    \end{array}
    \right)
    \, , \qquad 
    \det \mathbf{P} = (-1)^{k_1 k_2}
    \, .
\end{align}
Defining
\begin{align}
    \Pi = 
    \operatorname{diag}
    \left(
        \mathbf{P}, \mathbf{P}, \ldots, \mathbf{P}
    \right)
    \, ,
    \qquad
    \det \Pi = (-1)^{N k_1 k_2}
    \, ,
\end{align}
we have
\begin{align}
     \det a D 
     & = \det (a D \Pi) \det \Pi^{-1}
     \nonumber \\
     & =
    \begin{vmatrix}
    C & A & 0 &&& 0 & Y \\
    B & -C^\dag & -\id_{k_2} & 0 && 0 & 0 \\
    0 & - \id_{k_1} & C & A & 0 & & \\
    & 0 & B & -C^\dag & -\id_{k_2} & \ddots & \\
    &&\ddots&\ddots&\ddots&\ddots& 0 \\
    0 & 0 &&& - \id_{k_1} & C & A \\
    X & 0 &&& 0 & B & -C^\dag
    \end{vmatrix}
    \times (-1)^{N k_1 k_2}
    \nonumber \\
     & =
    \begin{vmatrix}
     A & 0 &&&& Y & C \\
     -C^\dag & -\id_{k_2} & 0 &&& 0 & B\\
     - \id_{k_1} & C & A & 0 &&& \\
     0 & B & -C^\dag & -\id_{k_2} & 0 && \\
    &\ddots&\ddots&\ddots&\ddots&\ddots& \\
     &&& -\id_{k_1} & C & A & 0 \\
     &&& 0 & B & -C^\dag & X
    \end{vmatrix}
    \times (-1)^{(N-1) k_2^2}
    \, .
\end{align}
We then define the following matrices,
\begin{align}
    \alpha = 
    \begin{pmatrix}
        A & 0 \\ - C^\dag & - \id_{k_2}
    \end{pmatrix}
    \, , \quad 
    \tilde\alpha = 
    \begin{pmatrix}
        A & 0 \\ - C^\dag & X
    \end{pmatrix}    
    \, , \quad 
    \beta = 
    \begin{pmatrix}
        - \id_{k_1} & C \\ 0 & B
    \end{pmatrix}
    \, , \quad 
    \tilde\beta = 
    \begin{pmatrix}
        Y & C \\ 0 & B
    \end{pmatrix}
    \, ,
\end{align}
from which we deduce a simple form,
\begin{align}
    \det aD & = 
    \begin{vmatrix}
        \alpha & & & \tilde\beta \\
        \beta & \alpha & & \\
        &\ddots&\ddots& \\
        &&\beta& \tilde\alpha
    \end{vmatrix}
    \times (-1)^{(N-1) k_2^2}
    \, .
\end{align}
Noticing that $\det \alpha = (-1)^{k_2} \det A$, we evaluate the Dirac determinant as follows,
\begin{align}
    \det a D 
    & = (-1)^{(N-1) k_2^2} \det \alpha^N \det \left( \alpha^{-1} \tilde{\alpha} - (- \alpha^{-1} \beta)^N \beta^{-1} \tilde{\beta} \right)
    \nonumber \\
    & = (-1)^{n} \det A^N \det \left( 
    \begin{pmatrix}
        \id_{k_1} & 0 \\ 0 & -X
    \end{pmatrix}
    - T^{-N}
    \begin{pmatrix}
        -Y & 0 \\ 0 & \id_{k_2}
    \end{pmatrix}
    \right)
\end{align}
where $n = (N-1) k_2^2 + N k_2$ and the $T$-operator is defined in Definition~\ref{def:T-op}.
This is the expression shown in~\eqref{eq:Dirac_det_classA}.
\hfill 
$\square$

\bibliographystyle{JHEP}
\bibliography{ref}

\providecommand{\href}[2]{#2}\begingroup\raggedright\begin{thebibliography}{10}

\bibitem{Thouless:1982zz}
D.J.~Thouless, M.~Kohmoto, M.P.~Nightingale and M.~den Nijs, \emph{{Quantized
  Hall Conductance in a Two-Dimensional Periodic Potential}},
  \href{https://doi.org/10.1103/PhysRevLett.49.405}{\emph{Phys. Rev. Lett.}
  {\bfseries 49} (1982) 405}.

\bibitem{Kohmoto:1985AP}
M.~Kohmoto, \emph{Topological invariant and the quantization of the hall
  conductance}, \href{https://doi.org/10.1016/0003-4916(85)90148-4}{\emph{Ann.
  Phys.} {\bfseries 160} (1985) 343}.

\bibitem{Niu:1984uz}
Q.~Niu, D.J.~Thouless and Y.-S.~Wu, \emph{{Quantized Hall Conductance as a
  Topological Invariant}},
  \href{https://doi.org/10.1103/PhysRevB.31.3372}{\emph{Phys. Rev. B}
  {\bfseries 31} (1985) 3372}.

\bibitem{Schnyder:2008tya}
A.~Schnyder, S.~Ryu, A.~Furusaki and A.~Ludwig, \emph{{Classification of
  topological insulators and superconductors in three spatial dimensions}},
  \href{https://doi.org/10.1103/PhysRevB.78.195125}{\emph{Phys. Rev. B}
  {\bfseries 78} (2008) 195125}
  [\href{https://arxiv.org/abs/0803.2786}{{\ttfamily 0803.2786}}].

\bibitem{Kitaev:2009mg}
A.~Kitaev, \emph{{Periodic table for topological insulators and
  superconductors}}, \href{https://doi.org/10.1063/1.3149495}{\emph{AIP Conf.
  Proc.} {\bfseries 1134} (2009) 22}
  [\href{https://arxiv.org/abs/0901.2686}{{\ttfamily 0901.2686}}].

\bibitem{Altland:1997zz}
A.~Altland and M.R.~Zirnbauer, \emph{{Nonstandard symmetry classes in
  mesoscopic normal-superconducting hybrid structures}},
  \href{https://doi.org/10.1103/PhysRevB.55.1142}{\emph{Phys. Rev. B}
  {\bfseries 55} (1997) 1142}
  [\href{https://arxiv.org/abs/cond-mat/9602137}{{\ttfamily
  cond-mat/9602137}}].

\bibitem{Neuberger:1997fp}
H.~Neuberger, \emph{{Exactly massless quarks on the lattice}},
  \href{https://doi.org/10.1016/S0370-2693(97)01368-3}{\emph{Phys. Lett. B}
  {\bfseries 417} (1998) 141}
  [\href{https://arxiv.org/abs/hep-lat/9707022}{{\ttfamily hep-lat/9707022}}].

\bibitem{Neuberger:1997bg}
H.~Neuberger, \emph{{Vector - like gauge theories with almost massless fermions
  on the lattice}}, \href{https://doi.org/10.1103/PhysRevD.57.5417}{\emph{Phys.
  Rev. D} {\bfseries 57} (1998) 5417}
  [\href{https://arxiv.org/abs/hep-lat/9710089}{{\ttfamily hep-lat/9710089}}].

\bibitem{Neuberger:1998wv}
H.~Neuberger, \emph{{More about exactly massless quarks on the lattice}},
  \href{https://doi.org/10.1016/S0370-2693(98)00355-4}{\emph{Phys. Lett. B}
  {\bfseries 427} (1998) 353}
  [\href{https://arxiv.org/abs/hep-lat/9801031}{{\ttfamily hep-lat/9801031}}].

\bibitem{Hasenfratz:1998ri}
P.~Hasenfratz, V.~Laliena and F.~Niedermayer, \emph{{The Index theorem in QCD
  with a finite cutoff}},
  \href{https://doi.org/10.1016/S0370-2693(98)00315-3}{\emph{Phys. Lett. B}
  {\bfseries 427} (1998) 125}
  [\href{https://arxiv.org/abs/hep-lat/9801021}{{\ttfamily hep-lat/9801021}}].

\bibitem{Luscher:1998pqa}
M.~Lüscher, \emph{{Exact chiral symmetry on the lattice and the
  Ginsparg-Wilson relation}},
  \href{https://doi.org/10.1016/S0370-2693(98)00423-7}{\emph{Phys. Lett. B}
  {\bfseries 428} (1998) 342}
  [\href{https://arxiv.org/abs/hep-lat/9802011}{{\ttfamily hep-lat/9802011}}].

\bibitem{Ginsparg:1981bj}
P.H.~Ginsparg and K.G.~Wilson, \emph{{A Remnant of Chiral Symmetry on the
  Lattice}}, \href{https://doi.org/10.1103/PhysRevD.25.2649}{\emph{Phys. Rev.
  D} {\bfseries 25} (1982) 2649}.

\bibitem{Hasenfratz:1993sp}
P.~Hasenfratz and F.~Niedermayer, \emph{{Perfect lattice action for
  asymptotically free theories}},
  \href{https://doi.org/10.1016/0550-3213(94)90261-5}{\emph{Nucl. Phys. B}
  {\bfseries 414} (1994) 785}
  [\href{https://arxiv.org/abs/hep-lat/9308004}{{\ttfamily hep-lat/9308004}}].

\bibitem{Bietenholz:2000ca}
W.~Bietenholz and J.~Nishimura, \emph{{Ginsparg-Wilson fermions in odd
  dimensions}},
  \href{https://doi.org/10.1088/1126-6708/2001/07/015}{\emph{JHEP} {\bfseries
  07} (2001) 015} [\href{https://arxiv.org/abs/hep-lat/0012020}{{\ttfamily
  hep-lat/0012020}}].

\bibitem{Huet:1996pw}
P.Y.~Huet, R.~Narayanan and H.~Neuberger, \emph{{Overlap formulation of
  Majorana-Weyl fermions}},
  \href{https://doi.org/10.1016/0370-2693(96)00443-1}{\emph{Phys. Lett. B}
  {\bfseries 380} (1996) 291}
  [\href{https://arxiv.org/abs/hep-th/9602176}{{\ttfamily hep-th/9602176}}].

\bibitem{Narayanan:1996mr}
R.~Narayanan and H.~Neuberger, \emph{{Overlap for Majorana-Weyl fermions}},
  \href{https://doi.org/10.1016/S0920-5632(96)00746-3}{\emph{Nucl. Phys. B
  Proc. Suppl.} {\bfseries 53} (1997) 658}
  [\href{https://arxiv.org/abs/hep-lat/9607080}{{\ttfamily hep-lat/9607080}}].

\bibitem{Inagaki:2004ar}
T.~Inagaki and H.~Suzuki, \emph{{Majorana and Majorana-Weyl fermions in lattice
  gauge theory}},
  \href{https://doi.org/10.1088/1126-6708/2004/07/038}{\emph{JHEP} {\bfseries
  07} (2004) 038} [\href{https://arxiv.org/abs/hep-lat/0406026}{{\ttfamily
  hep-lat/0406026}}].

\bibitem{Suzuki:2004ht}
H.~Suzuki, \emph{{A No-go theorem for the Majorana fermion on the lattice}},
  \href{https://doi.org/10.1143/PTP.112.855}{\emph{Prog. Theor. Phys.}
  {\bfseries 112} (2004) 855}
  [\href{https://arxiv.org/abs/hep-lat/0407010}{{\ttfamily hep-lat/0407010}}].

\bibitem{Hayakawa:2006fd}
M.~Hayakawa and H.~Suzuki, \emph{{Gauge anomaly associated to the Majorana
  fermion in 8k+1 dimensions}},
  \href{https://doi.org/10.1143/PTP.115.1129}{\emph{Prog. Theor. Phys.}
  {\bfseries 115} (2006) 1129}
  [\href{https://arxiv.org/abs/hep-th/0601026}{{\ttfamily hep-th/0601026}}].

\bibitem{Fukui:2009pc}
T.~Fukui and T.~Fujiwara, \emph{{A Z$_2$ index of Dirac operator with time
  reversal symmetry}},
  \href{https://doi.org/10.1088/1751-8113/42/36/362003}{\emph{J. Phys. A}
  {\bfseries 42} (2009) 362003}
  [\href{https://arxiv.org/abs/0905.3639}{{\ttfamily 0905.3639}}].

\bibitem{Ringel:2012fm}
Z.~Ringel and A.~Stern, \emph{{$\mathbb{Z}_2$ anomaly and boundaries of
  topological insulators}},
  \href{https://doi.org/10.1103/PhysRevB.88.115307}{\emph{Phys. Rev. B}
  {\bfseries 88} (2013) 115307}
  [\href{https://arxiv.org/abs/1212.3796}{{\ttfamily 1212.3796}}].

\bibitem{Kimura:2015ixh}
T.~Kimura, \emph{{Domain-wall, overlap, and topological insulators}},
  \href{https://doi.org/10.22323/1.251.0042}{\emph{PoS} {\bfseries LATTICE2015}
  (2016) 042} [\href{https://arxiv.org/abs/1511.08286}{{\ttfamily
  1511.08286}}].

\bibitem{Ando:2015sia}
Y.~Ando and L.~Fu, \emph{{Topological Crystalline Insulators and Topological
  Superconductors: From Concepts to Materials}},
  \href{https://doi.org/10.1146/annurev-conmatphys-031214-014501}{\emph{Ann.
  Rev. Condensed Matter Phys.} {\bfseries 6} (2015) 361}
  [\href{https://arxiv.org/abs/1501.00531}{{\ttfamily 1501.00531}}].

\bibitem{Adams:1998eg}
D.H.~Adams, \emph{{Axial anomaly and topological charge in lattice gauge theory
  with overlap Dirac operator}},
  \href{https://doi.org/10.1006/aphy.2001.6209}{\emph{Annals Phys.} {\bfseries
  296} (2002) 131} [\href{https://arxiv.org/abs/hep-lat/9812003}{{\ttfamily
  hep-lat/9812003}}].

\bibitem{Fukaya:2020tjk}
H.~Fukaya, M.~Furuta, Y.~Matsuki, S.~Matsuo, T.~Onogi, S.~Yamaguchi et~al.,
  \emph{{Mod-two APS index and domain-wall fermion}},
  \href{https://doi.org/10.1007/s11005-022-01509-2}{\emph{Lett. Math. Phys.}
  {\bfseries 112} (2022) 16}
  [\href{https://arxiv.org/abs/2012.03543}{{\ttfamily 2012.03543}}].

\bibitem{Yamashita:2020nkf}
M.~Yamashita, \emph{{A Lattice Version of the Atiyah\textendash{}Singer Index
  Theorem}}, \href{https://doi.org/10.1007/s00220-021-04021-1}{\emph{Commun.
  Math. Phys.} {\bfseries 385} (2021) 495}
  [\href{https://arxiv.org/abs/2007.06239}{{\ttfamily 2007.06239}}].

\bibitem{Kubota:2020tpr}
Y.~Kubota, \emph{{The Index Theorem of Lattice Wilson\textendash{}Dirac
  Operators via Higher Index Theory}},
  \href{https://doi.org/10.1007/s00023-022-01159-z}{\emph{Annales Henri
  Poincare} {\bfseries 23} (2022) 1297}
  [\href{https://arxiv.org/abs/2009.03570}{{\ttfamily 2009.03570}}].

\bibitem{Clancy:2023ino}
M.~Clancy, D.B.~Kaplan and H.~Singh, \emph{{Generalized Ginsparg-Wilson
  equations}},  \href{https://arxiv.org/abs/2309.08542}{{\ttfamily
  2309.08542}}.

\bibitem{Lee:2019ole}
J.Y.~Lee, J.~Ahn, H.~Zhou and A.~Vishwanath, \emph{{Topological Correspondence
  between Hermitian and Non-Hermitian Systems: Anomalous Dynamics}},
  \href{https://doi.org/10.1103/PhysRevLett.123.206404}{\emph{Phys. Rev. Lett.}
  {\bfseries 123} (2019) 206404}
  [\href{https://arxiv.org/abs/1906.08782}{{\ttfamily 1906.08782}}].

\bibitem{Kawabata:2018gjv}
K.~Kawabata, K.~Shiozaki, M.~Ueda and M.~Sato, \emph{{Symmetry and Topology in
  Non-Hermitian Physics}},
  \href{https://doi.org/10.1103/PhysRevX.9.041015}{\emph{Phys. Rev. X}
  {\bfseries 9} (2019) 041015}
  [\href{https://arxiv.org/abs/1812.09133}{{\ttfamily 1812.09133}}].

\bibitem{Sun:2018PRL}
X.-Q.~Sun, M.~Xiao, T.~Bzdu{\v{s}}ek, S.-C.~Zhang and S.~Fan,
  \emph{{Three-Dimensional Chiral Lattice Fermion in Floquet Systems}},
  \href{https://doi.org/10.1103/physrevlett.121.196401}{\emph{Phys. Rev. Lett.}
  {\bfseries 121} (2018) } [\href{https://arxiv.org/abs/1806.09296}{{\ttfamily
  1806.09296}}].

\bibitem{Bessho:2020hrs}
T.~Bessho and M.~Sato, \emph{{Nielsen-Ninomiya Theorem with Bulk Topology:
  Duality in Floquet and Non-Hermitian Systems}},
  \href{https://doi.org/10.1103/PhysRevLett.127.196404}{\emph{Phys. Rev. Lett.}
  {\bfseries 127} (2021) 196404}
  [\href{https://arxiv.org/abs/2006.04204}{{\ttfamily 2006.04204}}].

\bibitem{Fukaya:2021sea}
H.~Fukaya, \emph{{Understanding the index theorems with massive fermions}},
  \href{https://doi.org/10.1142/S0217751X21300155}{\emph{Int. J. Mod. Phys. A}
  {\bfseries 36} (2021) 2130015}
  [\href{https://arxiv.org/abs/2109.11147}{{\ttfamily 2109.11147}}].

\bibitem{Fujikawa:1998if}
K.~Fujikawa, \emph{{A Continuum limit of the chiral Jacobian in lattice gauge
  theory}}, \href{https://doi.org/10.1016/S0550-3213(99)00042-5}{\emph{Nucl.
  Phys. B} {\bfseries 546} (1999) 480}
  [\href{https://arxiv.org/abs/hep-th/9811235}{{\ttfamily hep-th/9811235}}].

\bibitem{Suzuki:1998yz}
H.~Suzuki, \emph{{Simple evaluation of chiral Jacobian with overlap Dirac
  operator}}, \href{https://doi.org/10.1143/PTP.102.141}{\emph{Prog. Theor.
  Phys.} {\bfseries 102} (1999) 141}
  [\href{https://arxiv.org/abs/hep-th/9812019}{{\ttfamily hep-th/9812019}}].

\bibitem{Kikukawa:1998pd}
Y.~Kikukawa and A.~Yamada, \emph{{Weak coupling expansion of massless QCD with
  a Ginsparg-Wilson fermion and axial U(1) anomaly}},
  \href{https://doi.org/10.1016/S0370-2693(99)00021-0}{\emph{Phys. Lett. B}
  {\bfseries 448} (1999) 265}
  [\href{https://arxiv.org/abs/hep-lat/9806013}{{\ttfamily hep-lat/9806013}}].

\end{thebibliography}\endgroup

\end{document}